\documentclass[11pt,a4paper]{article}
\usepackage{latexsym,amssymb,amsmath,amsthm,amsfonts,enumerate,verbatim,xspace,
exscale}

\input xy
\xyoption{all} \CompileMatrices \UseComputerModernTips

\usepackage{graphicx,tabularx}
\usepackage{amsmath,amsbsy,amsfonts,amssymb}
\usepackage{color}
\usepackage{amsthm}
\usepackage{mathtools}
\usepackage{lineno}

\title{Momentum map reduction for nonholonomic systems}

\date{}

\author{P. Balseiro  \footnotemark}
\author{{\sc{Paula Balseiro}\thanks{
         Universidade Federal Fluminense, Instituto de Matem\'atica e Estat\'\i stica, Rua Mario Santos Braga S/N, 24020-140,
        Niteroi, Rio de Janeiro, Brazil.}
        }  \ \
        {\sc{Maria Eugenia Garcia}\thanks{CMaLP - Centro de Matem\'atica de La Plata, Universidad Nacional de La Plata,  Calles 50 y 115, La Plata, Argentina.}
\thanks{Facultad de Ciencias Exactas, Universidad Nacional de La Plata, Calles 50 y 115, La Plata, Argentina. } } \ \ 
{\sc{Cora Tori}\footnotemark[2]
\thanks{Facultad de Ingenier\'ia, Universidad Nacional de La Plata, Calles 48 y 116, La Plata, Argentina. } } \ \ 
{\sc{Marcela Zuccalli}\footnotemark[2] \footnotemark[3]
 } 
}

\theoremstyle{plain}
\newtheorem{theorem}{Theorem}[section]
\newtheorem{lemma}[theorem]{Lemma}
\newtheorem{proposition}[theorem]{Proposition}

\newtheorem*{theorem*}{Theorem}
\newtheorem{remarkth}[theorem]{Remark}

\theoremstyle{definition}
\newtheorem{definition}[theorem]{Definition}

\newenvironment{remark}{\begin{remarkth}\upshape}{\hfill$\diamond$\end{remarkth}}

\newcommand{\g}{\mathfrak{g}}

\def\W{\mathcal{W}}
\def\M{\mathcal{M}}

\def\S{\mathcal{S}}
\def\C{\mathcal{C}}

\def\R{\mathbb{R}}

\def\red{{\mbox{\tiny{red}}}}
\def\nh{{\mbox{\tiny{nh}}}}

\def\B{{\mbox{\tiny{$B$}}}}
\def\BB{{\mbox{\tiny{${\bf B}$}}}}
\def\subGW{{\mbox{\tiny{$G_{\!W}$}}}}
\def\subW{{\mbox{\tiny{$W$}}}}

\def\subS{{\mbox{\tiny{$S$}}}}
\def\subM{{\mbox{\tiny{$\M$}}}}
\def\subQ{{\mbox{\tiny{$Q$}}}}

\def\vecOm{\boldsymbol{\Omega}}

\newcommand{\SO}{\mbox{$\textup{SO}$}}
\newcommand{\SE}{\mbox{$\textup{SE}$}}

\def\vecL{\boldsymbol{\lambda}}

\def\vecgamma{\boldsymbol{\gamma}}
\def\vecalpha{\boldsymbol{\alpha}}
\def\vecbeta{\boldsymbol{\beta}}

\def\vecs{\boldsymbol{s}}
\def\vecx{{\bf x}}

\begin{document}

\maketitle

\begin{abstract}

This paper presents a reduction procedure for nonholonomic systems admitting suitable types of symmetries and conserved quantities. The full procedure contains two steps. The first (simple) step results in a Chaplygin system, described by an almost symplectic structure, carrying additional symmetries.  The focus of this paper is on the second step, which consists of a Marsden-Weinstein--type reduction that generalizes constructions in  \cite{BalFern, Cortes:Book}. The almost symplectic manifolds obtained in the second step are proven to coincide with the leaves of the reduced nonholonomic brackets defined in \cite{BY}. We illustrate our construction with several classical examples.

\end{abstract}

\section{Introduction}

A basic fact in symplectic geometry, widely used in geometric mechanics, is that a symplectic manifold carrying a (free and proper) symplectic action gives rise to a Poisson bracket on the quotient manifold. 
Moreover, if the action is Hamiltonian then the Marsden-Weinstein reduced spaces of the symplectic manifold, at different values of the momentum map, coincide with unions of symplectic leaves of the quotient Poisson structure. Given an invariant Hamiltonian function, the dynamics in the quotient Poisson manifold restrict to leaves and hence can be studied by means of Marsden-Weinstein reduction. This paper presents some analogs of these results in the context of nonholonomic systems.

The study of nonholonomic systems with symmetries has a vast literature, see e.g. \cite{Bloch:Book, Cortes:Book, Book:CDS}. In our set-up, a nonholonomic system is geometrically described by an {\em almost} Poisson 
structure (the lack of integrability being a consequence of the constraints in velocities \cite{BS93,IdLMdD}) along with a Hamiltonian function. In the presence of symmetries, it is shown in \cite{BY,GNM} that, if the system admits suitable conserved quantities (called {\em horizontal gauge momenta} \cite{BGM,FGS-08}), then there is a modification of the almost Poisson bracket that still codifies the nonholonomic dynamics and has the following key property:  the corresponding reduced bracket on the quotient manifold, though generally not Poisson, gives rise to a foliation by almost symplectic leaves that are tangent to the reduced nonholonomic vector field. Our goal in this paper is to study, in this context, a Marsden-Weinstein--type reduction that produces these almost symplectic leaves.
This procedure extends the ones in  \cite{BalFern,BS93,Cortes:Book} in that we allow for more general conserved quantities as well as modifications of the almost symplectic structure (by  {\it dynamical gauge transformations}) prior to reduction.

We now explain the framework and results in this article in more detail. A nonholonomic system is determined by a configuration manifold $Q$, a Lagrangian $L$ and a non-integrable distribution $D$ on $Q$ describing the permitted velocities. The submanifold $\M \subset T^*Q$ given by the image of the Legendre transformation of $D$ has a natural almost Poisson bracket  \cite{IdLMdD,Marle,SchaftMaschke1994}, called the {\it nonholonomic bracket}, and a Hamiltonian function $H_\subM$ defined by $L$. The nonholonomic dynamics on $\M$ is determined by the ``hamiltonian'' vector field of $H_\subM$ with respect to the nonholonomic bracket, denoted by 
$X_\nh$. If the nonholonomic system has symmetries given by the (free and proper) action of a Lie group $G$, then the nonholonomic bracket and the dynamics can be reduced to the quotient manifold  $\M/G$.


In our set-up, we assume that $G$ admits a closed normal subgroup $G_\subW$ so that the nonholonomic system is $G_\subW$-Chaplygin \cite{Koiller}. A consequence of this fact is that, setting $\widetilde Q:=Q/G_\subW$,
the nonholonomic vector field descends to a vector field $\widetilde X_{\nh}$ on the cotangent bundle $T^*\widetilde Q$,  which is the hamiltonian vector field of the reduced hamiltonian function $\widetilde H$ with respect to a natural {\em almost} symplectic 2-form $\widetilde \Omega$ on $T^*\widetilde Q$ (coming from the nonholonomic bracket and defined in detail in Sec.~\ref{Ss:VerticalSymmetries}).  This is the first step in our reduction procedure.

 Our goal in this paper is to explain a further reduction of the nonholonomic system $(T^*\widetilde Q, \widetilde \Omega, \widetilde H)$ making use of the canonical momentum map for the action of the remaining Lie group $F:=G/G_\subW$ on the cotangent bundle $T^*\widetilde Q$. Here a first difficulty is that the vector field $\widetilde X_{\nh}$ is not tangent to the momentum level sets. We fix this problem by using conserved quantities of the system to find suitable $F$-invariant submanifolds that substitute the momentum level sets in the reduction procedure.  A second difficulty is that $\widetilde \Omega$ is not basic on these $F$-invariant submanifolds. This issue is resolved through a suitable modification of $\widetilde \Omega$ by a special 2-form. We elaborate on these two key points below.

 \smallskip 

 \noindent {\em The $F$-invariant submanifolds carrying the nonholonomic dynamics}. A central assumption in this work is the existence of the maximum possible amount of certain types of first integrals --horizontal gauge momenta-- defined by the evaluation of the canonical momentum map $\widetilde{J}$ on $T^*\widetilde Q$ on given {\em $\mathfrak{f}$-valued functions} $\eta_i$ on $\widetilde Q$, for $i=1,\ldots, k $, where $\mathfrak{f}$ is the Lie algebra of $F$.
We show that, for $\mathfrak{f}^*$-valued functions $\mu = \sum_i c_i \mu^i$, where $\mu^i$'s are dual to the $\eta_i$'s and $c_i \in \mathbb{R}$, we obtain $F$-invariant submanifolds 
$$
\widetilde J^{-1}(\mu):=\{ \alpha_x \in T^*\widetilde Q\,|\, \widetilde{J}(\alpha_x)=\mu(x) \} \subset T^*\widetilde Q
$$ 
which are $F$-invariant and foliate $T^*\widetilde Q$ in such a way that $\widetilde X_\nh$ is always tangent to them, see Sec.~\ref{Ss:Reduction}, Prop.~\ref{P:LevelSets}.

\smallskip

\noindent{\em The modification of $\widetilde \Omega$}. What is behind the fact that the pull-back of $\widetilde\Omega$ to $\widetilde J^{-1}(\mu)$ does not descend to  the quotient $\widetilde J^{-1}(\mu)/F$ is that the infinitesimal generator\footnote{The infinitesimal generator of a $\mathfrak{f}$-valued function $\eta$ is defined at $\alpha_x \in T^*_x\widetilde{Q}$ as the infinitesimal generator of $\eta(x)\in \mathfrak{f}$.} $(\eta_i)_{T^*\widetilde{Q}}$ of $\eta_i$  is not necessarily the ``hamiltonian'' vector field associated to the horizontal gauge momentum $\widetilde J_{\eta_i}$.  Following \cite{BY}, we define a 2-form $\widetilde{\bf B}$ on $T^*\widetilde Q$ that satisfies 
\begin{equation}\label{Eq:Intro}
{\bf i}_{(\eta_i)_{T^*\widetilde Q}} (\widetilde \Omega + \widetilde{\bf B}) = d\widetilde J_{\eta_i},
\end{equation}
as well as the {\it dynamical condition} ${\bf i}_{\widetilde X_{\nh}}\widetilde{\bf B} =0$. Note that, by this last condition,
our nonholonomic system is equivalently described by the triple $(T^*\widetilde Q,\widetilde \Omega + \widetilde{\bf B}, \widetilde H)$. We prove in Theorem \ref{T:MW} that the pull-back of the 2-form $\widetilde \Omega+\widetilde{\bf B}$ to the manifold $J^{-1}(\mu)$ is basic and hence defines an almost symplectic form $\omega^\BB_\mu$ on $\widetilde J^{-1}(\mu)/F$. 

\smallskip

Let us stress that an important point in our construction is that, in general, the $\mu$'s are suitable $\mathfrak{f}^*$-valued {\em functions},
not just fixed elements of $\mathfrak{f}^*$ as in the usual hamiltonian case. This is essential for the reduction to be compatible with the nonholonomic dynamics. Comparing with previous constructions, we
note that in \cite{BalFern,BS93} the conserved quantities are assumed to be defined by fixed elements in the Lie algebra, while in \cite{Cortes:Book} the reduction procedure considers $\mathfrak{f}$-valued function but, due to the lack of the 2-form $\widetilde{\bf B}$, it was not possible to define a reduced 2-form on the quotients $\widetilde J^{-1}(\mu)/F$. 

The 2-form $\widetilde{\bf B}$ was defined in \cite{BY,GNM} in the context of hamiltonization, and its explicit expression permits a better understanding of the resulting ``Marsden-Weinstein'' reduced spaces even in the specific cases studied in previous works. In particular, inspired by the hamiltonian case \cite{AbrMarsden,Marsdenetal} and  using  the {\it shift-trick}, we show that the  almost symplectic manifolds $(\widetilde J^{-1}(\mu)/F, \omega_\mu^\BB)$ are diffeomorphic to the  manifold $T^*(\widetilde Q/F)$ with its canonical symplectic 2-form modified by a term $\widehat{\mathcal{B}}_\mu$ that only depends on the 2-form $\widetilde{\bf B}$, see Theorem~\ref{T:mu-identification}.

In Sec.~\ref{S:NHBracket},  we relate the almost symplectic reduced spaces obtained in our construction
with an almost Poisson bracket on the $\M/G$ given by the reduction of a modification of the non-holonomic bracket on $\M$ considered in \cite{BY,GNM}.  As shown in these papers,  when a nonholonomic system  admits the maximum amount of horizontal gauge momenta, the {\it gauge transformation} of the nonholonomic bracket on $\M$ by a suitable 2-form ${\bf B}$ generates a new bracket  whose reduction by symmetries gives an almost Poisson bracket $\{\cdot, \cdot\}_{\red}^\BB$ on $\M/G$ that admits an almost symplectic foliation. We show in Theorem \ref{T:LeavesOfBracket} that its leaves agree with the connected components of the almost symplectic reduced spaces of Theorem.~\ref{T:MW}. Having a Marsden-Weinstein--type description of the almost symplectic foliation associated to the reduced bracket $\{\cdot, \cdot\}_\red^\BB$ is useful to study the dynamics restricted to leaves, to find conformal factors for the reduced brackets $\{\cdot, \cdot\}_{\red}^\BB$, as well as to study Routh reduction, integrability, Hamilton-Jacobi theory and even numerical methods (e.g. variational integrators), see \cite{Bol-Bor-Ma-11,Bol-Kil-Kaz-14,Cortes:Book,Cortes-Martinez-01,dL-MdD-SM-04,FGS-05,FTZ,MaRaSche,Oscar-11}.

Besides the Chaplygin ball (that was also treated in \cite{BalFern}), in Sec.~\ref{S:Examples} we study many other examples that could not be treated in \cite{BalFern,BS93,Cortes:Book}, starting from the simple example of the nonholonomic particle, the snakeboard \cite{BalSan, BKMM} and the more sophisticated one describing a solid of revolution rolling on plane \cite{Book:CDS, BalSolids,GNM}.

{\bf Acknowledgment:}  P.B. thanks to Facultad de Ciencias Exactas UNLP (Argentina) for the financial support and the hospitality during her visits.

\tableofcontents

\section{Nonholonomic systems and first step reduction}\label{S:Prelim}

In this section we will define the basic concepts around nonholonomic systems with symmetries and, in particular, the vertical symmetry condition which permits the reduction in two steps. 

\subsection{Nonholonomic systems with symmetries}\label{Ss:NHSys}

A nonholonomic system is a mechanical system on a manifold $Q$ with a lagrangian function $L:TQ\to \R$ and (linear) constraints in the velocities.  The permitted velocities define a (constant rank) nonintegrable distribution $D$ on $Q$.  
Throughout this paper we assume that the lagrangian $L$ is of mechanical type: $L= \frac{1}{2}\kappa - U$ where $\kappa$ is the kinetic energy metric and $U$ the potential.

Next, we write the nonholonomic equations of motion in the hamiltonian framework following \cite{BS93}.  The Legendre transformation $\kappa^\sharp :TQ\to T^*Q$ given, at $X,Y\in TQ$, by $\kappa^\sharp(X)(Y): = \kappa(X,Y)$, defines the submanifold $\M$ of $T^*Q$ by $\M:=\kappa^\sharp(D)$.  Since the Legendre transformation is linear on the fibers, then $\tau_\subM := \tau|_\M :\M\to Q$ is a vector subbundle of the canonical vector bundle $\tau: T^*Q\to Q$.
The nonintegrable distribution $D$ induces a (nonintegrable) distribution $\C$ on $\M$, with fiber at each $m\in \M$, given by 
\begin{equation}\label{Eq:C}
\C_m :=\{ v_m \in T_m\M \ : \ T\tau_\subM(v_m) \in D_q, \ \ \mbox{for}\ q = \tau_\subM(m) \in Q\}.
\end{equation}
Let $H: T^*Q\to \R$ be the hamiltonian function associated to the lagrangian $L$ and $\Omega_\subQ$ the canonical 2-form on $T^*Q$. Considering $\iota:\M\to T^*Q$ the natural inclusion, we denote by $\Omega_\subM:=\iota^*\Omega_\subQ$ and $H_\subM:=\iota^*H$ the pull backs of $\Omega_\subQ$ and $H$ to the submanifold $\M$, respectively. 
Following \cite{BS93}, the nonholonomic dynamics is described by the integral curves of the vector field $X_\nh$ --called the {\it nonholonomic vector field}-- defined on $\M$ given by 
$$
{\bf i}_{X_\nh} \Omega_\subM|_\C = dH_\subM|_\C,
$$
where $(\cdot)|_\C$ is the point-wise restriction to $\C$. During this paper, we will use the triple $(\M, \Omega_\subM|_\C, H_\subM)$ to define a nonholonomic system.

Consider a free and proper action of a Lie group $G$ on $Q$.  This action is a {\it symmetry} of the nonholonomic system if its tangent lift leaves the lagrangian $L$ and the distribution $D$ invariant, or equivalently, if the cotangent lift of the action leaves $\M$ and $H$ invariant. Therefore, there is a well defined $G$-action on the manifold $\M$ denoted by $\varPsi:G\times\M \to \M$.  

If the nonholonomic system admits a $G$-symmetry, the nonholonomic vector field $X_\nh$ is $G$-invariant as well: $T\varPsi_g(X_\nh(m)) = X_\nh(\varPsi_g(m))$ for all $m\in\M$, $g\in G$. Thus, the vector field $X_\nh$ descends to the quotient manifold $\M/G$, defining the \emph{reduced nonholonomic vector field} $X_\red$ given by 
$$
X_\red := T\rho(X_\nh),
$$ 
where $\rho:\M\to \M/G$ is the orbit projection.

In what follows, we will consider symmetries of the nonholonomic system satisfying the so-called {\it dimension assumption} \cite{BKMM}, that is, for each $q\in Q$,
$$
T_q Q= D_q + V_q,
$$
where $V_q$ is the tangent space to the $G$-orbit on $Q$ at $q$. As usual, we denote by $S$ the distribution on $Q$ defined, for each $q\in Q$, by $S_q:=D_q\cap V_q$. Since the $G$-action is free, then both distributions $S$ and $V$ have constant rank. 

Let $\g$ be the Lie algebra associated to the Lie group $G$ and consider the trivial bundle $Q\times\g\to Q$
whose sections can be thought as $\g$-valued functions,  that is, if $\xi\in \Gamma(Q\times\g)$, at each $q\in Q$, then $\xi_q = \xi(q)\in \g$.
Then, the distribution $S$ induces the subbundle $\g_S\to Q$ of $Q\times\g \to Q$ with fiber, at $q\in Q$, given by
$$
(\g_S)_q:=\{ \xi_q\in \g \ : \ (\xi_q)_\subQ (q) \in S_q\},
$$
where $(\xi_q)_\subQ(q)$ is the infinitesimal generator of the element $\xi_q\in \g$ at $q$ (see e.g. \cite{BKMM}).  For short, we may denote by $\xi_\subQ(q):= (\xi_q)_\subQ(q)$. The bundle $\g_S\to Q$ has the same rank as the distribution $S$, i.e., 
$$
k:=\textup{rank}(S) = \textup{rank}(\g_S).
$$
Equivalently, the dimension assumption can be written, for each $m\in \M$, as $T_m\M = \C_m + \mathcal{V}_m$ where $\mathcal{V}_m$ is the tangent to the orbit associated to the $G$-action on $\M$ at $m$. Analogously, the distribution $\mathcal{S}$ on $\M$ is defined, at each $m\in \M$, by $\mathcal{S}_m := \C_m \cap \mathcal{V}_m$. Therefore, if $\xi\in\Gamma(\g_S)$ then its infinitesimal generator on $\M$ satisfies that $\xi_\subM(m)\in \mathcal{S}_m$ (in this case, $\xi_\subM(m) := (\xi_q)_\subM (m)$, where $q=\tau_\subM(m)$).

It is well known that, even in the presence of symmetries, the canonical momentum map does not generate conserved functions because it does not take into account the constraints. In order to consider the constraints, the {\it nonholonomic momentum map} was defined in \cite{BKMM} as the bundle map $J^\nh:\M\to \g_S^*$ given, at each $m\in\M$ and $\xi\in \Gamma(\g_S)$ by 
$$
\langle J^\nh(m), \xi(q)\rangle := {\bf i}_{\xi_\subM} \Theta_\subM(m),
$$
where $q=\tau_\subM(m)$ and $\Theta_\subM :=\iota^*\Theta_\subQ$, recalling that $\iota:\M\to T^*Q$ is the natural inclusion and $\Theta_\subQ$ is the Liuoville 1-form on $T^*Q$. Observe also that $J^\nh$ is the (pull back to $\M$ of the) canonical momentum map on $(T^*Q,\Omega_\subQ)$ but evaluated on $\g_S$-valued functions on $Q$. It was also studied in \cite{BKMM} a momentum equation, involving a PDE, encoding the conservation of functions of the type $J_\xi := \langle J^\nh, \xi\rangle$, where $\langle J^\nh, \xi\rangle(m) = \langle J^\nh(m), \xi(q)\rangle$.  In fact, when such a function $J_\xi$ is conserved by the nonholonomic dynamics, i.e., $X_\nh(J_\xi)=0$, is called {\it horizontal gauge momentum} \cite{BGM} and the associated element $\xi\in \Gamma(\g_S)$ is a {\it horizontal gauge symmetry}. 

The original definition of {\it horizontal gauge momenta} was done in local coordinates and independently of the nonholonomic momentum map. 

\begin{remark}
The general existence of horizontal gauge momenta is still an open problem and what is usually done is to assume their existence when it is needed (for more details about their properties and existence see \cite{BalSan,FGS-05,FGS-08}). 
\end{remark}

\begin{lemma} \label{L:J_xi-Invariant}
The function $J_\xi = \langle J^{\emph\nh} , \xi\rangle$ is $G$-invariant on $\M$ if and only if the section $\xi$ on $\g_\subS\to Q$ is $Ad$-invariant: that is for $q\in Q$ and $g\in G$,    $Ad_g \, ( \xi(\Psi_{g^{-1}}(q))) = \xi(q)$, where $\Psi_g:Q\to Q$ is the $G$-action on $Q$. 
\end{lemma}
\begin{proof}
The function $J_\xi= \langle J^\nh,\xi\rangle$ is $G$-invariant if and only if $J_\xi(m) = J_\xi(\varPsi_g(m))$ which means that, for $m_q\in \M_q\subset T_q^*Q$,  $\langle m_q, (\xi_q)_\subQ\rangle = \langle \Psi^*_{g^{-1}} (m_q) , (\xi(\Psi_g(q)))_\subQ(\Psi_g(q)\rangle = \langle m_q, T\Psi_{g^{-1}}( (\xi(\Psi_g(q)))_\subQ (\Psi_g(q)) )\rangle = \langle m_q, (Ad_{g^{-1}} (\xi(\Psi_g(q))) )_\subQ (q)\rangle$. Therefore, $\xi(q) = Ad_{g^{-1}} ( \xi(\Psi_g(q)) )$. 
  
\end{proof}

\subsection{The vertical symmetry condition} \label{Ss:VerticalSymmetries}

Let $(\M, \Omega_\subM|_\C, H_\subM)$ be a nonholonomic system with a $G$-symmetry satisfying the dimension assumption.
We say that a distribution $W$ is a {\it vertical complement of the constraints} $D$ if
\begin{equation}\label{Eq:TQ=D+W}
TQ = D\oplus W \qquad \mbox{and} \qquad W \subset V.
\end{equation}
Due to the dimension assumption, a vertical complement of the constraints always exists but is not uniquely defined. 
The choice of a vertical complement $W$ induces a splitting of the vertical space 
$$
V = S\oplus W,
$$
and, consequently,  a splitting of the bundle $Q\times\g\to Q$ so that $Q\times\g = \g_\subS \oplus \g_\subW$, where $\g_\subW\to Q$ is the subbundle of $Q\times\g \to Q$ with fibers
$$
(\g_\subW)_q :=\{ \xi_q\in \g \ : \ (\xi_q)_\subQ (q) \in W_q\}.
$$

During this paper, we will assume that the Lie group $G$ admits a closed normal subgroup $G_\subW$  so that the system is Chaplygin with respect to the $G_\subW$-action. Hence, the system can be reduced in two steps: first by $G_\subW$ and subsequently by the Lie group $F=G/G_\subW$.  More precisely,

\begin{proposition}\label{P:VertSym-Equivalencies}\emph{\cite{DuistermaatKolk}}
Let $W$ be a vertical complement of the constraints $D$ so that $\g_\subW \simeq Q\times\mathfrak{w}$ for $\mathfrak{w}$ a Lie subalgebra of $\g$.  Then the following assertions are equivalent:
\begin{enumerate} \setlength\itemsep{-0.3em}
 \item[$(i)$] $W$ is $G$-invariant (or equivalently the Lie algebra $\mathfrak{w}$ is Ad-invariant),
 \item[$(ii)$] $\mathfrak{w}$ is an ideal of $\g$,
 \item[$(iii)$] there exist a normal subgroup $G_\subW$ of $G$ whose Lie algebra is $\mathfrak{w}$,
 \item[$(iv)$] for each $q\in Q$,  $W_q = T_q (\textup{Orb}_{G_W}(q))$.
\end{enumerate}
\end{proposition}


\begin{definition}\label{Def:VerticalSymmetries}\cite{BJac}
A vertical complement $W$ of the constraints $D$ satisfies the {\it vertical symmetry condition} if there  exists a closed normal subgroup $G_\subW$ of $G$ so that for each $q\in Q$,  $W_q = T_q (\textup{Orb}_{G_W}(q))$. 
\end{definition}

The vertical symmetry condition implies that the bundle $\g_\subW$ is trivial and moreover $\g_\subW \simeq Q\times \mathfrak{w}$ where $\mathfrak{w} \subset \g$ is the Lie algebra of the Lie group $G_\subW$.  

Asking a vertical complement $W$ to satisfy the vertical symmetry condition is a restrictive requisite, however there are many examples admitting such type of complement, for instance the nonholonomic oscillator, the snakeboard, the Chaplygin ball and the solids of revolution (see Sec.~\ref{S:Examples}).

Let us consider a nonholonomic system $(\M,\Omega_\subM|_\C, H_\subM)$ with a $G$-symmetry satisfying the dimension assumption. 
If the vertical complement $W$ satisfies the vertical symmetry condition, then the Lie group $G_\subW$ acts freely and properly on $Q$ and its tangent lift leaves the lagrangian $L$ and the constraint distribution $D$ invariant. Therefore the nonholonomic system has also symmetries given by the action of the Lie group $G_\subW$ and thus the nonholonomic vector field $X_\nh$ on $\M$ descends to a vector field $\widetilde X_\nh$ on the quotient manifold $\widetilde \M := \M/G_\subW$ so that 
\begin{equation}\label{Eq:TildeDyn}
\widetilde X_\nh := T\rho_{\mbox{\tiny{$G_{\!W}$}}} (X_\nh),
\end{equation}
where $\rho_{\mbox{\tiny{$G_{\!W}$}}} : \M \to \widetilde \M$ is the $G_\subW$-orbit projection.
Moreover, since $TQ = D\oplus W$, where $W$ is the tangent to the $G_\subW$-orbit on $Q$ (see Prop.~\ref{P:VertSym-Equivalencies}), then we see that the nonholonomic system is a {\it $G_\subW$-Chaplygin system}, see e.g. \cite{Koiller}. 
In fact, in \cite{Koiller}, it is proven that $\widetilde \M$ is diffeomorphic to the cotangent manifold $T^*\widetilde Q$, for $\widetilde Q := Q/G_\subW$, and that the partially reduced dynamics $\widetilde X_\nh$ is hamiltonian with respect to an almost symplectic 2-form.  More precisely, if $\widetilde H$ is the (partially) reduced  hamiltonian function  on $T^*\widetilde Q$ such that $\rho_{\mbox{\tiny{$G_{\!W}$}}}^*\widetilde H = H_\subM$ and  $\Omega_{\mbox{\tiny{$\widetilde Q$}}}$ is the canonical symplectic 2-form on $T^*\widetilde Q$, then
 \begin{equation} \label{Eq:DynChaplygin}
{\bf i}_{\widetilde X_{\nh}} \widetilde \Omega = d\widetilde H, \qquad \mbox{with} \qquad 
\widetilde \Omega:= \Omega_{\mbox{\tiny{$\widetilde Q$}}} -B_{\mbox{\tiny{$\!\langle\! J\mathcal{K}\!\rangle$}}},\end{equation}
where $B_{\mbox{\tiny{$\!\langle\! J\mathcal{K}\!\rangle$}}}$ is the 2-form on $T^*\widetilde Q$ defined as follows: the splitting \eqref{Eq:TQ=D+W} and the vertical symmetry condition, induces a principal connection $A_\subW :TQ\to \mathfrak{w}$ given, at each $v_q\in T_qQ$ by 
$A_\subW(v_q) = \eta_q,$ if and only if $P_\subW(v_q) = (\eta_q)_\subQ$, where $P_\subW:TQ\to W$ is the projection to the second  factor of \eqref{Eq:TQ=D+W}.
We define the $\mathfrak{w}$-valued 2-form $K_\subW$ on $Q$ given by $K_\subW := d^DA_\subW$, that is, for $X,Y\in TQ$, 
\begin{equation}\label{Eq:difD}
K_\subW(X,Y) = d^DA_\subW(X,Y) := dA_\subW(P_D(X),P_D(Y)),
\end{equation}
where $P_D:TQ\to D$ is the projection to the first factor of \eqref{Eq:TQ=D+W}.  Since $\mathfrak{w}\subset \g$, we can see $K_\subW$ as a $\g$-valued 2-form on $Q$ and consider its pull back to $\M$, i.e., $\mathcal{K}_\subW := \tau_\subM^*K_\subW$.  Then, following \cite{Koiller}, we define the 2-form $\langle J, \mathcal{K}_\subW\rangle$ as the natural paring of the canonical momentum map $J:\M\to \g^*$ with the $\g$-valued 2-form $\mathcal{K}_\subW$ (where $\langle \cdot, \cdot \rangle$ denotes the pairing between $\g^*$ and $\g$).  
The 2-form $\langle J, \mathcal{K}_\subW\rangle$ was proven to be basic with respect to the principal bundle $\rho_{\mbox{\tiny{$G_{\!W}$}}} : \M \to  T^*\widetilde Q$ and therefore $B_{\mbox{\tiny{$\!\langle\! J\mathcal{K}\!\rangle$}}}$ is the 2-form on $T^*\widetilde Q$ such that $\rho_{\mbox{\tiny{$G_{\!W}$}}}^*B_{\mbox{\tiny{$\!\langle\! J\mathcal{K}\!\rangle$}}} = \langle J, \mathcal{K}_\subW\rangle$.

\subsection{Chaplygin systems with an extra symmetry}

From Section~\ref{Ss:VerticalSymmetries} and in particular from \eqref{Eq:DynChaplygin}, we consider the partially reduced nonholonomic system $(T^*\widetilde Q, \widetilde \Omega, \widetilde H)$.  

Since the Lie group $G_\subW$ is a closed normal subgroup of $G$, the quotient $F := G/G_\subW$ is a Lie group with Lie algebra $\mathfrak{f} := \mathfrak{g}/\mathfrak{w}$ where $k=\textup{dim}(\mathfrak{f})$. In this section we study the $F$-symmetry on the manifold $(T^*\widetilde Q, \widetilde\Omega)$.     

Let us denote by $\varrho_{\mbox{\tiny{$G$}}}:G\to F$ the projection to the quotient Lie group and by $\varrho_{\mbox{\tiny{$\g$}}} : \g \to \mathfrak{f}$ the projection to the quotient Lie algebra.  
If $\Psi:G\times Q\to Q$ denotes the $G$-action on $Q$, then the (partially) reduced manifold $\widetilde Q$ inherits a well defined action of the Lie group $F$ denoted by $\widetilde \Psi : F \times \widetilde Q \to \widetilde Q$, given, at each $h\in F$ and $x\in\widetilde Q$, by
\begin{equation}\label{Eq:F-action}
\widetilde \Psi_h (x) := \rho_{\widetilde\subQ} (\Psi_g(q)), 
\end{equation}
where $g\in G$ and $q\in Q$ satisfy that $\varrho_{\mbox{\tiny{$G$}}}(g) = h$ and $\rho_{\widetilde\subQ} (q) = x$, respectively,  for $\rho_{\widetilde\subQ}:Q \to \widetilde Q$ the orbit projection.
Therefore, the Lie group $F$ acts (freely and properly) also on the manifold $T^*\widetilde Q$ leaving the 2-form $\widetilde \Omega$ and the hamiltonian $\widetilde H$ invariant.

The following Lemma will be useful to study the horizontal gauge momenta on the partially reduced manifold $T^*\widetilde Q$.  Recall that a section $\xi$ of $\g_\subS\to Q$ is  $Ad$-invariant if, as seen as $\g$-valued functions on $Q$, we have that for $q\in Q$ and $g\in G$,    $Ad_g \, ( \xi(\Psi_{g^{-1}}(q))) = \xi(q)$.

\begin{lemma}\label{L:lemma1-1} Consider the bundles $\g_\subS\to Q$ and $\widetilde Q\times \mathfrak{f}\to \widetilde Q$, and the projection to the orbits $\rho_{\widetilde\subQ}: Q\to \widetilde Q$.  
 \begin{enumerate}   \item[$(i)$] There is a one-to-one correspondence between the $Ad$-invariant sections on $\g_\subS\to Q$ and sections of $\widetilde Q\times \mathfrak{f}\to \widetilde Q$ so that if $\xi\in \Gamma(\g_S)$ then there is a unique $\eta\in \Gamma(\widetilde Q\times \mathfrak{f})$ such that
 $T\rho_{\widetilde\subQ}(\xi_\subQ(q)) = \eta_{\mbox{\tiny{$\widetilde Q$}}}(x)$.
  
   \item[$(ii)$] The choice of a horizontal $G$-invariant distribution $H\!\emph{\mbox{\tiny{or}}} \subset D$ such that $TQ=H\!\emph{\mbox{\tiny{or}}} \oplus V$, induces a $F$-invariant splitting on $T\widetilde Q=\widetilde{H\!\emph{\mbox{\tiny{or}}}}  \oplus \widetilde V_{\mbox{\tiny{F}}}$ where $\widetilde V_{\mbox{\tiny{F}}}$ is the tangent to the $F$-orbit on $\widetilde Q$ and $\widetilde {H\!\emph{\mbox{\tiny{or}}}}:=T\rho_{\mbox{\tiny{$\widetilde Q$}}}(H\!\emph{\mbox{\tiny{or}}})$.  In other words, an equivariant principal connection $A:TQ\to \g$, induces an equivariant principal connection $\widetilde{A}:T\widetilde Q \to \mathfrak{f}$ on $\widetilde Q$ so that the following diagram commutes:
 $$
 \xymatrix{TQ \ar[r]^A \ar[d]_{T\rho_{\mbox{\tiny{$\widetilde Q$}}}}  & \g \ar[d]^{\varrho_\g} \\
 T \widetilde Q \ar[r]^{\widetilde A} & \mathfrak{f} }
 $$

 \end{enumerate}

\end{lemma}

\begin{proof} $(i)$ First, let $\xi\in \Gamma(\g_S)$ $Ad$-invariant.  Then, the vector field $\xi_\subQ$ is invariant and we can define the vector field $X$ on $\widetilde Q$ such that, for each $x=\rho_{\widetilde\subQ}(q) \in \widetilde Q$,  $X(x) = T\rho_{\widetilde\subQ}(\xi_\subQ(q))$.  Now we claim that $X(x)\in V_{\mbox{\tiny{$F$}}}$ where $V_{\mbox{\tiny{$F$}}}$ is the tangent to the $F$-orbit on $\widetilde Q$.  In fact, using that $\rho = \rho_{\mbox{\tiny{$\widetilde\subQ/F$}}} \circ \rho_{\widetilde\subQ}$ for $\rho_{\mbox{\tiny{$\widetilde\subQ/F$}}} : \widetilde Q\to \widetilde Q/F$ the orbit projection, we have that $0 = T\rho(\xi_\subQ(q)) = T\rho_{\mbox{\tiny{$\widetilde\subQ/F$}}}(X(x))$ and hence $X(x)\in (V_{\mbox{\tiny{$F$}}})_x$.  Therefore, for each $x\in \widetilde Q$ there is $\eta(x)\in \mathfrak{f}$ such that $ T\rho_{\widetilde\subQ}(\xi_\subQ(q)) = X(x) =  \eta_{\widetilde\subQ}(x)$.

Conversely, if $\eta\in \Gamma(\widetilde Q\times \mathfrak{f})$, then $\eta_{\widetilde\subQ}(x)\in T_x\widetilde Q$. Therefore there exists an (unique) invariant vector field $Y$ on $Q$ such that $T\rho_{\widetilde\subQ}(Y(q))= \eta_{\widetilde\subQ}(x)$ and $Y(q)\in D_q$ (recall that $TQ= D\oplus W$ and $\textup{Ker}T\rho_{\widetilde\subQ} = W$).  Then we claim that, at each $q\in Q$, $Y(q)\in S_q$ since, using again that $\rho = \rho_{\mbox{\tiny{$\widetilde\subQ/F$}}} \circ \rho_{\widetilde\subQ}$, we have that 
$$
T\rho(Y(q)) = T\rho_{\mbox{\tiny{$\widetilde\subQ/F$}}} (T\rho_{\widetilde\subQ}(Y(q)) = T\rho_{\mbox{\tiny{$\widetilde\subQ/F$}}} (\eta_{\widetilde\subQ}(x)) = 0
$$
Therefore, $Y(q) \in V_q$ and hence $Y(q)\in S_q$.  Then, for each $q\in Q$, there is an element $\xi(q)\in \g_\subS|_q$ such that $Y(q) = \xi_\subQ(q)$. Using that the vector field $Y(q)$ is invariant, we obtain that $\xi\in \Gamma(\g_\subS)$ is $Ad$-invariant. 

 $(ii)$  Item $(i)$ asserts that $\textup{rank}(S) = \textup{rank}(\widetilde V_{\mbox{\tiny{$F$}}})$ and moreover, since $\textup{Ker}\, T\rho_{\mbox{\tiny{$\widetilde Q$}}} = W$, we have that $\textup{rank}(H\!\mbox{\tiny{or}})=\textup{rank}(\widetilde {H\!\mbox{\tiny{or}}})$ and hence we conclude that $T\widetilde Q = \widetilde {H\!\mbox{\tiny{or}}}\oplus \widetilde V_{\mbox{\tiny{$F$}}}$. It is straightforward to see that $\widetilde A \circ T\rho_{\mbox{\tiny{$\widetilde Q$}}} = \varrho_\g \circ A$. 

\end{proof}

\subsubsection*{The conserved quantity assumption}

Consider a nonholonomic system $(\M, \Omega_\subM|_\C, H_\subM)$ with a $G$-symmetry and recall that $S=D\cap V$. 
Next we will make a fundamental assumption that will be used the rest of the paper:
{\it the nonholonomic system $(\M, \Omega_\subM|_\C, H_\subM)$ admits $k=\textup{rank}(S)$ $G$-invariant (functionally independent) horizontal gauge momenta $\{J_1,...,J_k\}$}. 
Since the corresponding horizontal gauge symmetries $\zeta_i\in \Gamma(\g_\subS)$, such that $J_i := \langle J^\nh , \zeta_i\rangle$, are linearly independent and globally defined,  they define a global basis of ($Ad$-invariant) sections of $\g_\subS\to Q$ denoted by 
\begin{equation}\label{Eq:BasisHGS}
\mathfrak{B}_{\mbox{\tiny{HGS}}}= \{\zeta_1, ...,\zeta_k\}.
\end{equation}
As a consequence of Lemma \ref{L:lemma1-1}, the global basis $\mathfrak{B}_{\mbox{\tiny{HGS}}}$ induces a corresponding basis of global sections 
\begin{equation}\label{Eq:RedBasisHGS}
\widetilde{\mathfrak{B}}_{\mbox{\tiny{HGS}}}  = \{\eta_1,...,\eta_k\},
\end{equation}
of the bundle $\widetilde{Q}\times \mathfrak{f}\to \widetilde{Q}$ where, for each $i = 1,...,k$, $(\eta_i)(x):= \varrho_{\mbox{\tiny{$\g$}}}( (\xi_i)(q))$ for $q\in Q$ and $x = \rho_{\widetilde\subQ}(q)\in \widetilde Q$.  We will often see the elements $\eta_i\in  \widetilde{\mathfrak{B}}_{\mbox{\tiny{HGS}}}$ as $\mathfrak{f}$-valued functions on $\widetilde Q$.
Associated to the basis $\widetilde{\mathfrak{B}}_{\mbox{\tiny{HGS}}}$ we can define the functions $\{\widetilde J_1,..., \widetilde J_k\}$ on $T^*\widetilde Q$, given by
\begin{equation}\label{Eq:TildeJi}
\widetilde J_i:= {\bf i}_{(\eta_i)_{T^*\widetilde{Q}}} \Theta_{\mbox{\tiny{$\widetilde Q$}}},
\end{equation}
where $\Theta_{\widetilde\subQ}$ is the canonical 1-form on $T^*\widetilde Q$.

\begin{proposition}\label{Prop:tildeJ_i} Recalling that $\rho_\subGW : \M\to T^*\widetilde Q$ is the orbit projection, we have that
\begin{enumerate}
 \item[$(i)$] the functions $\widetilde J_i$ on $T^*\widetilde Q$ are functionally independent and $\rho_\subGW^*  {\widetilde J}_i = J_i$,
 \item[$(ii)$] the functions $\widetilde J_i$ are conserved by the partially reduced dynamics $\widetilde X_{\emph\nh}$.
\end{enumerate}
\end{proposition}

\begin{proof}
 $(i)$ Let $\xi_i$ be the $Ad$-invariant horizontal gauge symmetry in \eqref{Eq:BasisHGS} and $\eta_i$ be the corresponding $\mathfrak{f}$-valued functions on $\widetilde Q$ defined in \eqref{Eq:RedBasisHGS}.  Using \eqref{Eq:TildeJi} and that $(\rho^*_\subGW \Theta_{\mbox{\tiny{$\widetilde Q$}}} - \Theta_\subM)|_\C = 0$ we obtain that 
 $$
 \rho^*_\subGW (\widetilde J_i) = \rho_\subGW^* ({\bf i}_{(\eta_i)_{T^*\!\widetilde Q}} \Theta_{\mbox{\tiny{$\widetilde Q$}}}) = {\bf i}_{(\xi_i)_\subM}\, \rho_\subGW^*\Theta_{\mbox{\tiny{$\widetilde Q$}}} = {\bf i}_{(\xi_i)_\subM}\Theta_\subM = J_i.
 $$
 
 $(ii)$ It is a consequence of item $(i)$. 
\end{proof}

\begin{definition} \label{Def:RedHorGauge}
The functions $\widetilde J_i={\bf i}_{(\eta_i)_{T^*\widetilde Q}} \Theta_{\widetilde\subQ}$ for $\eta_i\in \widetilde{\mathfrak{B}}_{\mbox{\tiny{HGS}}}$,
are called the {\it partially reduced horizontal gauge momenta} and the corresponding $\mathfrak{f}$-valued functions $\eta_i$ are the {\it partially reduced horizontal gauge symmetries}.
\end{definition}

We conclude that, if the nonholonomic system $(\M, \Omega_\subM|_\C , H_\subM)$ admits $k$ (functionally independent) $G$-invariant horizontal gauge momenta, then the partially reduced system $(T^*\widetilde Q, \Omega_{\mbox{\tiny{$\widetilde Q$}}}, \widetilde H)$ inherits $k=\textup{dim}(\mathfrak{f})$ partially reduced horizontal gauge momenta.  

\medskip

From Lemma \ref{L:lemma1-1}$(ii)$ the vector fields $(\eta_i)_{\widetilde\subQ}$ generate the vertical space $\widetilde{V}_{\mbox{\tiny{$F$}}}$ and hence the connection $\widetilde A: T\widetilde Q\to \mathfrak{f}$ can be written as 
 \begin{equation}\label{Eq:TildeA} 
\widetilde A  = \widetilde Y^i \otimes \eta_i,  
 \end{equation}
 for $\eta_i\in \widetilde{\mathfrak{B}}_{\mbox{\tiny{HGS}}}$ and $\widetilde Y^i$ the 1-forms on $\widetilde Q$ so that $\widetilde Y^i((\eta_j)_{\mbox{\tiny{$\widetilde Q$}}}) = \delta_{ij}$ and $\widetilde Y^i|_{\widetilde {H\!{\mbox{\tiny{or}}}}} = 0$.

\section{Momentum map reduction} \label{S:Reduction}

In this section, we will work with the partially reduced system $(T^*\widetilde Q, \widetilde \Omega, \widetilde H)$ defined in \eqref{Eq:DynChaplygin} and the corresponding symmetry group $F$. We assume the existence of $k = \textup{dim}(\mathfrak{f})$ partially reduced horizontal gauge momenta $\{\widetilde J_1, ...,\widetilde J_k\}$ (Def.~\ref{Def:RedHorGauge}) with corresponding partially reduced horizontal gauge symmetries given by $\widetilde{\mathfrak{B}}_{\mbox{\tiny{HGS}}}  = \{\eta_1,...\eta_k\}$ as in \eqref{Eq:RedBasisHGS}. With these ingredients we will define an almost symplectic foliation --where the reduced nonholonomic dynamics lives-- trough a Marsden-Weinstein-type reduction. 

\subsection{The canonical momentum map}\label{Ss:CanMomentMap}

The canonical momentum map $\widetilde J : T^*\widetilde Q\to \mathfrak{f}^*$ on  $(T^*\widetilde Q,\Omega_{\widetilde\subQ} )$ is defined, as usual, by
\begin{equation}\label{Eq:MomMapTilde}
\langle  \widetilde J (\alpha_x), \nu \rangle = {\bf i}_{\nu_{T^*\!\widetilde Q}} \Theta_{\mbox{\tiny{$\widetilde Q$}}}{\mbox{\tiny{$(\alpha_x)$}}},
\end{equation}
  for $\alpha_x \in T_x^*\widetilde{Q}$ and  $\nu  \in \mathfrak{f}$. Then, for each $\nu\in\mathfrak{f}$, the function $\widetilde J_\nu\in C^\infty (T^*\widetilde Q)$ is given, at each $\alpha_x\in T^*_x\widetilde Q$, by $\widetilde J_\nu (\alpha_x) = \langle  \widetilde J , \nu \rangle(\alpha_x) :=\langle  \widetilde J (\alpha_x), \nu \rangle$.  However, these functions do not see the partially reduced horizontal gauge momenta $\widetilde J_i$ since the associated $\mathfrak{f}$-valued functions $\eta_i \in \widetilde{\mathfrak{B}}_{\mbox{\tiny{HGS}}}$  are not necessarily constant (c.f.\cite{BalFern}).

In order to encode the partially reduced horizontal gauge momenta using the canonical momentum map, for each $\mathfrak{f}$-valued function $\eta$, we define the function $\widetilde J_\eta$ on $T^*\widetilde Q$ by
\begin{equation}\label{Eq:J-functions}
\widetilde J_\eta(\alpha_x) = \langle  \widetilde J , \eta \rangle(\alpha_x) :=\langle  \widetilde J (\alpha_x), \eta(x) \rangle,
\end{equation}
for $\alpha_x\in T^*_x\widetilde Q$.
Therefore,  the partially reduced horizontal gauge momenta $\{ \widetilde J_1,...,\widetilde J_k\}$ given in Def.~\ref{Def:RedHorGauge} associated to the basis $\widetilde{\mathfrak{B}}_{\mbox{\tiny{HGS}}} = \{\eta_1,...,\eta_k\}$ given in \eqref{Eq:RedBasisHGS},  are described by the canonical momentum map using \eqref{Eq:J-functions} 
$$
\widetilde J_i := \widetilde J_{\eta_i} = \langle \widetilde J, \eta_i\rangle, \quad \mbox{for } i=1,...,k.
$$

We denote by 
\begin{equation}\label{Eq:RedDualBasisHGS}
\widetilde{\mathfrak{B}}^*_{\mbox{\tiny{HGS}}} = \{\mu^1,...,\mu^k\},
\end{equation}
the dual basis of $\mathfrak{f}^*$-valued functions on $\widetilde Q$ associated to $\widetilde{\mathfrak{B}}_{\mbox{\tiny{HGS}}}$, that is, for each $i=1,...,k$ the $\mu^i$ are $\mathfrak{f}^*$-valued functions on $\widetilde Q$  (or  sections of the bundle $\widetilde{Q}\times \mathfrak{f}^*\to \widetilde{Q}$), such that, at each $x\in\widetilde Q$, $(\mu^i)(x)\in \mathfrak{f}^*$ and $\langle \mu^i(x) , \eta_j(x) \rangle = \delta_{ij}$ for $\langle \cdot, \cdot\rangle$ the natural pairing between $\mathfrak{f}$ and $\mathfrak{f}^*$. 

\begin{remark}
 The canonical momentum map can also be written as $\widetilde J = \widetilde J_1 \mu^1 + ... + \widetilde J_k\mu^k$, where at each $\alpha_x\in T^*\widetilde Q$, $\widetilde J(\alpha_x) = \widetilde J_1(\alpha_x). \mu^1(x) + ... + \widetilde J_k(\alpha_x). \mu^k(x) \in \mathfrak{f}^*$. 
\end{remark}

Now, let us consider a $\mathfrak{f}^*$-valued function on $\widetilde Q$ given by  $\mu = c_i\mu^i$ for $c_i$ constants in $\R$ and define the level set
$$
\widetilde J^{-1}(\mu) := \{\alpha_x\in T^*\widetilde Q \ : \ \widetilde J(\alpha_x) = \mu(x)\}. 
$$

\begin{proposition}\label{P:LevelSets} Let $\mu = c_i\mu^i$ be a $\mathfrak{f}^*$-valued function for $c_i$ constants in $\R$ and $\mu^i\in \widetilde{\mathfrak{B}}^*_{\emph{\mbox{\tiny{HGS}}}}$.
Then the inverse image $\widetilde{J}^{-1}(\mu)\subset T^*\widetilde Q$ coincides with the common level sets of the (partially reduced) horizontal gauge momenta $\widetilde{J}_1,...,\widetilde{J}_k$ at $c_1,...,c_k$ respectively, i.e., 
	$$
	\widetilde{J}^{-1}(\mu) = \bigcap_i \widetilde{J}_i^{-1}(c_i),
	$$
	and hence $\widetilde J^{-1}(\mu) $ is a $F$-invariant submanifold of $T^*\widetilde Q$. Moreover, the collection  (of connected components) of the manifolds  $\widetilde J^{-1}(\mu)$ for $\mu \in \textup{span}_{\R} \{\mu^i\}$ defines a foliation of the manifold $T^*\widetilde Q$. 
	
\end{proposition}
\begin{proof} If $\alpha_x \in T^*\widetilde Q$ then $\widetilde J(\alpha_x) = \mu(x)$, is equivalent to
$\widetilde J_i(\alpha_x) = \langle \widetilde J(\alpha_x), \eta_i(x) \rangle = c_i$ for all $i=1,...,k$ and hence $\widetilde{J}^{-1}(\mu) = \cap_i \widetilde{J}_i^{-1}(c_i)$.

Let us consider the $F$-invariant submersion $\mathcal{J} := (\widetilde J_1,...,\widetilde J_k) : T^*\widetilde Q \to \R^k$.  Since for $c=(c_1,...,c_k)\in \R^k$, $\mathcal{J}^{-1}(c) = \widetilde J^{-1}(\mu)$, we conclude that, for each $\mu \in \textup{span}_{\R} \{\mu^i\}$, $\widetilde J^{-1}(\mu)$ is a $F$-invariant manifold  and the collection of connected components of $\widetilde J^{-1}(\mu)$ defines a foliation of $T^*\widetilde Q$. 

\end{proof}

As a consequence of the previous proposition, we conclude that the (partially reduced) nonholonomic vector field $\widetilde X_\nh$ is tangent to the manifolds $\widetilde J^{-1}(\mu)$ for $\mu= c_i\mu^i$.  

However, it is important to note that the map $\widetilde J$ does not behave as a momentum map on $(T^*\widetilde Q, \widetilde\Omega)$, not even in the coordinates given by the horizontal gauge symmetries, in the sense that the vector fields $(\eta_i)_{T^*Q}$ might not be hamiltonian vector fields associated to the functions $\widetilde J_i$ (i.e., ${\bf i}_{(\eta_i)_{T^*Q}}\widetilde \Omega$ can be different from $d\widetilde J_i$).  This observation has a fundamental reflect when we want to study a ``Marsden-Weinstein reduction'':  the pull back of $\widetilde \Omega$ to the manifold $\widetilde J^{-1}(\mu)$ is not basic with respect to the bundle $\widetilde J^{-1}(\mu)\to \widetilde J^{-1}(\mu)/F$.  To solve this problem, in the next section we consider a {\it gauge transformation by a 2-form} $B$ (as it was done in \cite{BGN,BY,GN2008}) so that we have the desired relation between $(\eta_i)_{\mbox{\tiny{$T^*\widetilde Q$}}}$ and the functions $\widetilde J_i$.

\subsection{The suitable dynamical gauge transformation ${\bf B}$} \label{Ss:B}

Consider the partially reduced nonholonomic system $(T^*\widetilde Q, \widetilde\Omega, \widetilde H)$ and a 2-form $\widetilde B$ on $T^*\widetilde Q$.  
Following \cite{SeveraWeinstein} a {\it gauge transformation by a 2-form $\widetilde B$} of the 2-form $\widetilde \Omega$  is just considering the 2-form $\widetilde \Omega + \widetilde B$ on $T^*\widetilde Q$. If the 2-form $\widetilde B$ is semi-basic with respect to the bundle $T^*\widetilde Q\to \widetilde Q$, then the manifold $(T^*\widetilde Q, \widetilde \Omega + \widetilde B)$ is almost symplectic. 

\begin{definition} \cite{BGN}  \label{Def:DynGauge}
A 2-form $\widetilde B$ on $T^*\widetilde Q$ induces a {\it dynamical gauge transformation} of $\widetilde \Omega$ if $\widetilde B$ is semi-basic with respect to the bundle $T^*\widetilde Q\to \widetilde Q$ and ${\bf i}_{\widetilde X_\nh}\widetilde B=0$.  
\end{definition}

\begin{remark}
 The original definition of gauge transformation was done on Dirac structures in \cite{SeveraWeinstein} and the {\it dynamical gauge transformation by a 2-form} was defined on almost Poisson structures in \cite{BGN}. In this section we work with the particular case of {\it dynamical gauge transformation by 2-forms} of 2-forms and in Sec.~\ref{S:NHBracket} we will see the relation with the corresponding almost Poisson brackets. 
\end{remark}

Definition \ref{Def:DynGauge} guarantees that the (partially reduced) dynamics $\widetilde X_\nh$ on $T^*\widetilde Q$ is also defined by 
$$
{\bf i}_{\widetilde X_\nh}(\widetilde \Omega+ \widetilde B)= d\widetilde H.
$$
The goal of considering this extra term given by a 2-form $\widetilde B$ is that there is a special choice of $\widetilde B$ such that the behaviour of the almost symplectic manifolds $(T^*\widetilde Q, \widetilde \Omega)$ and $(T^*\widetilde Q, \widetilde \Omega + B)$ are different regarding the (partially reduced) horizontal gauge momenta. That is, there is a 2-form $\widetilde{\bf B}$ on $T^*\widetilde Q$ that makes $(\eta_i)_{T^*\widetilde{Q}}$ the hamiltonian vector field of the functions $\widetilde J_i$ with respect to the 2-form $\widetilde \Omega + \widetilde{\bf B}$. 
Next, we write the explicit expression of such a 2-form $\widetilde{\bf B}$, which comes from a 2-form ${\bf B}$ defined on $\M$ presented in \cite{BY}.

Let $(\M, \Omega_\subM|_\C, H_\subM)$ be a nonholonomic system with a $G$-symmetry satisfying the dimension assumption.  Let $W$ be a vertical complement of the constraints \eqref{Eq:TQ=D+W} and choose a {\it horizontal} space $H\!\mbox{\tiny{or}} \subset D$ so that 
\begin{equation}\label{Eq:TQ=H+S+W} 
TQ = H\!\mbox{\tiny{or}}\oplus V = H\!\mbox{\tiny{or}}\oplus S \oplus W, 
\end{equation}
(observe that $D=H\!\mbox{\tiny{or}}\oplus S$).

We assume that the nonholonomic system $(\M, \Omega_\subM|_\C, H_\subM)$ admits $k=\textup{rank}(S)$ $G$-invariant (functionally independent) horizontal gauge momenta $\{J_1,...,J_k\}$. The corresponding horizontal gauge symmetries $\{\zeta_1,...,\zeta_k\}$ define the vector fields 
 $\{Y_1,...,Y_k\}$  on $Q$ given by $Y_i:=(\zeta_i)_\subQ$ and then we have the globally defined 1-forms $Y^i$ on $Q$ so that $Y^i|_{H\!\mbox{\tiny{or}}} = Y^i|_W = 0$ and $Y^i(Y_j) = \delta_{ij}$. 
Following \cite{BY}, we define the 2-form $B_1$ on $\M$ to be
\begin{equation}\label{Eq:B_1}
B_1 := \langle J, \mathcal{K}_\subW\rangle + J_i \, d^\C \mathcal{Y}^i
\end{equation} 
where, for  $i=1,...,k$, $\mathcal{Y}^i = \tau_\subM^*Y^i$ and $d^\C \mathcal{Y}^i = \tau_\subM^* d^DY^i$ with $d^DY^i (X,Y)= d{Y}^i (P_D(X),P_D(Y))$ for $X,Y\in TQ$ as in \eqref{Eq:difD}.

The splitting \eqref{Eq:TQ=H+S+W} induces a splitting on $T\M$ so that 
$$
T\M = \mathcal{H}\mbox{\tiny{or}} \oplus \mathcal{V} = \mathcal{H}\mbox{\tiny{or}} \oplus \S \oplus \W, 
$$
where $\mathcal{S}$ is defined in Sec.~\ref{Ss:NHSys}, $(\mathcal{H}\mbox{\tiny{or}})_m =\{v_m\in T_m\M \ : \ T\tau_\subM(v_m)\in (H\!\mbox{\tiny{or}})_q\}$ and $\mathcal{W}_m =\{v_m\in \mathcal{V}_m \ : \ T\tau_\subM(v_m)\in W_q\}$ at each $m\in \M$, $q=\tau_\subM(m)$.
Let us denote by $\mathcal{A} : T\M\to \g$ the principal connection with corresponding horizontal space $\mathcal{H}\mbox{\tiny{or}}$ and denote by $P_{\mbox{\tiny{$\mathcal{V}$}}}:\mathcal{H}\mbox{\tiny{or}}\oplus \mathcal{V} \to \mathcal{V}$ the projection to the second factor. Finally we also define the 2-form $\mathcal{B}$ on $\M$ as it was also done in \cite{BY}, 
\begin{equation}\label{Eq:CalB}
\mathcal{B} :=\langle J, \mathcal{K}_{\mbox{\tiny{$\mathcal{V}$}}} \rangle - \frac{1}{2}(\kappa_\g \wedge {\bf i}_{P_{\mbox{\tiny{$\mathcal{V}$}}}(X_\nh)}[\mathcal{K}_\subW +d^\C\mathcal{Y}^i\otimes \zeta_ i])|_{\mathcal{H}\mbox{\tiny{or}}},
\end{equation}
where $\mathcal{K}_{\mbox{\tiny{$\mathcal{V}$}}}$ is the curvature of $\mathcal{A}$,   and  $\kappa_\g$ is the $\g^*$-valued 1-form on $\M$ given, at each $\mathcal{X}\in T\M$, by $\kappa_\g (\mathcal{X}, \eta) = \kappa(T\tau_\M(\mathcal{X}), \eta_\subQ)$, for  $\eta\in\g$; for more details see \cite[Sec.~3.3]{BY}.

 \begin{proposition} \label{Prop:B_G-inv} \emph{\cite{BY}}
  The 2-forms $B_1$ and $\mathcal{B}$ defined in \eqref{Eq:B_1} and \eqref{Eq:CalB}, respectively, are semi-basic with respect to the bundle $\tau_\subM:\M\to Q$ and $G$-invariant. Moreover, the 2-form 
\begin{equation}\label{Def:B}
{\bf B} := B_1 + \mathcal{B},
\end{equation}
satisfies the {\it dynamical condition} 
\begin{equation} \label{Eq:DynamicCondition}
{\bf i}_{X_{\emph\nh}} {\bf B} = 0. 
\end{equation}
\end{proposition}

We assume now that the vertical complement $W$ verifies the vertical symmetry condition and hence, following Sec.~2, we consider the partially reduced system $(T^*\widetilde Q, \widetilde \Omega, \widetilde H)$ with the partially reduced horizontal gauge momenta $\{\widetilde J_1,...,\widetilde J_k\}$. 
Next, we will see that the 2-forms $B_1$ and $\mathcal{B}$ descend to well defined 2-forms on $T^*\widetilde Q$ and, in particular, $B_1$ has an explicit expression on $T^*\widetilde Q$.  

Following Lemma \ref{L:lemma1-1}$(ii)$  the splitting \eqref{Eq:TQ=H+S+W} induces the connection $\widetilde A = \widetilde Y^i\otimes \eta_i$ on $\widetilde Q$ as in \eqref{Eq:TildeA}. Define the 1-forms $\widetilde{\mathcal Y}^i$  on $T^*\widetilde Q$ such that $\tau_{\widetilde\subQ}^*\widetilde Y^i = \widetilde{\mathcal{Y}}^i$ for $\tau_{\widetilde\subQ} : T^*\widetilde Q \to \widetilde Q$ for each $i=1,...,k$ (or equivalently $\rho_\subGW^*\widetilde{\mathcal Y}^i = \mathcal{Y}^i$). Moreover, from Lemma \ref{L:lemma1-1} and \eqref{Eq:TildeA} these forms satisfy that $\widetilde{\mathcal Y}^i ( (\eta_i)_{\mbox{\tiny{$T^*\widetilde Q$}}}) = \delta_{ij}$.

\begin{proposition} \label{P:B-Chaplygin} Consider the nonholonomic system $(T^*\widetilde Q, \widetilde \Omega, \widetilde H)$,  
\begin{enumerate} \setlength\itemsep{-0.3em}
 \item[$(i)$] the 2-forms $B_1$ and $\mathcal{B}$ on $\M$ are basic with respect to the bundle $\rho_{\mbox{\tiny{$G_{\!W}$}}}:\M\to T^*\widetilde Q$, i.e., there exist $\widetilde{B}_1$ and $\widetilde{\mathcal B}$ on $T^*\widetilde Q$ such that $\rho_{\mbox{\tiny{$G_{\!W}$}}}^* \widetilde{B}_1 = {B}_1$ and $\rho_{\mbox{\tiny{$G_{\!W}$}}}^* \widetilde{\mathcal B} = {\mathcal B}$. In particular, 
 \begin{equation}\label{Eq:tildeB1}
\widetilde B_1 =B_{\mbox{\tiny{$\!\langle\! J\mathcal{K}\!\rangle$}}} + \widetilde J_i d\widetilde{\mathcal{Y}}^i,  
 \end{equation}
and $\widetilde{\mathcal{B}}$ is  basic with respect to the principal bundle $\rho_{\mbox{\tiny{F}}}:T^*\widetilde Q \to T^*\widetilde Q/F$, that is, there is a 2-form $\overline{\mathcal B}$ on $T^*\widetilde Q/F$ such that $\rho^*_{\mbox{\tiny{F}}}\overline{\mathcal B} = \widetilde{\mathcal{B}}$.
 
\item[$(ii)$] The 2-form  $\widetilde{\bf B} := \widetilde{B}_1 + \widetilde{\mathcal B}$ defines a dynamical gauge transformation on $(T^*\widetilde Q, \widetilde \Omega, \widetilde H)$.
 \end{enumerate}
\end{proposition}

\begin{proof}
 $(i)$ By construction, the 2-forms $B_1$ and ${\mathcal B}$ verify that ${\bf i}_X{B_1} = {\bf i}_X{\mathcal B}=0$ for all $X\in \Gamma(\W)$ (where $\mathcal{W}$ is the $G_\subW$-orbit on $\M$). Since they are  $G_\subW$-invariant (see Prop.~\ref{Prop:B_G-inv}),  we conclude that they are basic with respect to the bundle $\rho_{\mbox{\tiny{$G_{\!W}$}}}: \M\to T^*\widetilde Q$.  Therefore, from the expression of $B_1$ in \eqref{Eq:B_1} and by Prop.~\ref{Prop:tildeJ_i} $(i)$, it is straightforward to obtain that  $\widetilde B_1 =B_{\mbox{\tiny{$\!\langle\! J\mathcal{K}\!\rangle$}}} + \widetilde J_i d\widetilde{\mathcal{Y}}^i$, where $\widetilde{\mathcal Y}^i$ are 1-forms on $T^*\widetilde Q$ such that $\rho_\subGW^*\widetilde{\mathcal Y}^i = \mathcal{Y}^i$. Finally,  by \eqref{Eq:CalB} we see that $\mathcal{B}$ is semi-basic with respect to the bundle $\rho:\M\to \M/G$ and $G$-invariant and hence it is basic.   
 
$(ii)$ Since ${\bf B}$ is semi-basic with respect to the bundle $\tau_\subM: \M\to Q$ then $\widetilde{\bf B}$ is semi-basic with respect to $\tau_{\mbox{\tiny{$\widetilde{Q}$}}} : T^*\widetilde Q\to \widetilde Q$. The dynamical condition ${\bf i}_{\widetilde X_\nh} \widetilde{\bf B} = 0$ is a direct consequence of \eqref{Eq:DynamicCondition}.
\end{proof}
 
 As a consequence of Prop.~\ref{P:B-Chaplygin} and \eqref{Eq:DynChaplygin},  the (partially) reduced nonholonomic vector field $\widetilde X_\nh$ on $T^*\widetilde Q$ is also determined by
 $$
{\bf i}_{\widetilde X_\nh} \widetilde \Omega_\BB = d\widetilde H \qquad \mbox{where} \qquad \widetilde\Omega_\BB : = \widetilde\Omega + \widetilde{\bf B}.
$$
and, as a consequence,  the triple $(T^*\widetilde Q, \widetilde \Omega_\BB , \widetilde H)$ describes our (partially reduced) nonholonomic system as well.
 Note that, by Prop.~\ref{Prop:B_G-inv}, the 2-form $\widetilde{\bf B}$ is $F$-invariant and hence $F$ is a symmetry of the nonholonomic system $(T^*\widetilde Q, \widetilde \Omega_\BB , \widetilde H)$.

The following Proposition puts in evidence the need of considering the nonholonomic system des\-cribed by $(T^*\widetilde Q, \widetilde \Omega_\BB , \widetilde H)$ instead of $(T^*\widetilde Q, \widetilde \Omega, \widetilde H)$ showing that the canonical momentum map $\widetilde J: T^*\widetilde Q \to \mathfrak{f}^*$  is the map that behaves as a momentum on $(T^*\widetilde Q, \widetilde \Omega_\BB)$ when it is evaluated on the $\mathfrak{f}$-valued functions on $\widetilde Q$ of  $\widetilde{\mathfrak{B}}_{\mbox{\tiny{HGS}}}$ given in \eqref{Eq:RedBasisHGS}.

\begin{proposition}\label{Prop:HGM-MomMap} Consider the partially reduced nonholonomic system $(T^*\widetilde Q, \widetilde \Omega_\BB, \widetilde H)$.  The (partially reduced) horizontal gauge symmetries  $\eta_i\in \widetilde{\mathfrak{B}}_{\emph{\mbox{\tiny{HGS}}}}$ and  the momentum map $\widetilde J : T^*\widetilde Q\to \mathfrak{f}^*$ satisfy the relation
 $$
 {\bf i}_{(\eta_i)_{T^*\widetilde Q}} \widetilde \Omega_\BB = d \widetilde{J}_i,
 $$
 where $\widetilde J_i = \langle \widetilde J, \eta_i\rangle$ are the partially reduced horizontal gauge momenta.
\end{proposition}

\begin{proof}
Recall that $\widetilde \Omega_\BB = \widetilde \Omega + {\bf B} =\Omega_{\widetilde\subQ} - B_{\mbox{\tiny{$\!\langle\! J\mathcal{K}\!\rangle$}}}  + \widetilde B_1 + \widetilde{\mathcal{B}}$. Using \eqref{Eq:B_1} and the fact that ${\mathcal B}$ is basic with respect to $\rho_{\mbox{\tiny{$F$}}} : T^*\widetilde Q \to T^*\widetilde Q/F$, we have that 
$$
{\bf i}_{(\eta_i)_{T^*\widetilde Q}} \widetilde \Omega_\BB = {\bf i}_{(\eta_i)_{T^*\widetilde Q}} (\widetilde \Omega + \widetilde B_1 ) = {\bf i}_{(\eta_i)_{T^*\widetilde Q}} (\Omega_{\widetilde\subQ} +  \widetilde J_i d\widetilde{\mathcal{Y}}^i).
$$

Let us denote by $\widetilde Y_i:=(\eta_i)_{\mbox{\tiny{$\widetilde Q$}}}$. Consider now $G$-invariant vector fields $\widetilde X_1,..., \widetilde X_n$ so that \linebreak $\{\widetilde X_1,..., \widetilde X_n, \widetilde Y_1,..., \widetilde Y_k\}$ is a basis of vector fields on $\widetilde Q$ and consider its dual basis of 1-forms on $\widetilde Q$ given by $\{\widetilde X^1,..., X^n, \widetilde Y^1,..., \widetilde Y^k\}$. Since $\widetilde{\mathcal{Y}}^i  = \tau_{\widetilde\subQ}^*\widetilde Y^i$ for $\tau_{\widetilde\subQ} : T^*\widetilde Q \to \widetilde Q$ the canonical projection, we have that $\Theta_{\widetilde\subQ} = p_a \widetilde{\mathcal{X}}^a + p_i \widetilde{\mathcal{Y}}^i$ where $\widetilde{\mathcal{X}}^a  = \tau_{\widetilde\subQ}^*\widetilde X^a$ and $\widetilde J_i=p_i$ (see e.g. \cite{BY}), and hence we obtain that ${\bf i}_{(\eta_i)_{T^*\widetilde Q}} \widetilde \Omega_\BB = {\bf i}_{(\eta_i)_{T^*\widetilde Q}} (-d\Theta_{\widetilde \subQ} + \widetilde J_i d\widetilde{\mathcal{Y}}^i ) = d\widetilde J_i.$

\end{proof}

\begin{remark}
 
The particular case where the partially reduced symmetries are given by elements of the Lie algebra $\mathfrak{f}$  and the canonical momentum map behaves as a standard momentum map for $(T^*\widetilde Q, \widetilde \Omega_\BB , \widetilde H)$ was studied in \cite{BalFern} but we remark these are restrictive conditions, not satisfied in most of the cases as, for instance, the nonholonomic particle, the snakeboard and the solids of revolutions which are examples treated in Section \ref{S:Examples}. 
\end{remark}

\subsection{Almost symplectic reduction}\label{Ss:Reduction}

In this section we state one of the main results of the paper:  we will perform a reduction of $(T^*\widetilde Q, \widetilde \Omega_\BB)$ using the canonical momentum map $\widetilde J: T^*\widetilde Q\to \mathfrak{f}^*$ following the procedure of a Marsden-Weinstein--type reduction but having into account that the $\mathfrak{f}$-valued functions $\mu$ considered are, in general, non-constant functions on $\widetilde Q$.

\begin{theorem}\label{T:MW} 
Let $(\M,\Omega_\subM|_\C, H_\subM)$ be a nonholonomic system with a $G$-symmetry satisfying the dimension assumption.  Suppose that the system admits $\{J_1,...,J_k\}$, for $k=\textup{rank}(S)$,  $G$-invariant horizontal gauge momenta and that the vertical complement $W$ can be chosen so that it satisfies the vertical symmetry condition.  Then, for the 2-form  $\widetilde{\bf B}$ in Prop.~\ref{P:B-Chaplygin}, holds
 \begin{enumerate}
  \item[$(i)$] the partially reduced nonholonomic system $(T^*\widetilde Q, \widetilde \Omega_\BB, {\widetilde H})$ is $F$-invariant and it has $k$ (partially reduced) horizontal gauge momenta $\{\widetilde J_1,...,\widetilde J_k\}$.
  \item[$(ii)$] [Almost symplectic reduction] For each $\mathfrak{f}^*$-valued function of the form $\mu = c_i\mu^i$ for $c_i\in\R$ and $\mu^i\in \widetilde{\mathfrak{B}}^*_{\emph{\mbox{\tiny{HGS}}}}$ (defined in \eqref{Eq:RedDualBasisHGS}), the manifold $\widetilde J^{-1}(\mu)/F$ admits an almost symplectic form $\omega^\BB_\mu$ such that 
  $$
  \iota_\mu^* \widetilde \Omega_\BB = \rho_{\mbox{\tiny{$\mu$}}}^*\omega^\BB_\mu, 
  $$
  where $\iota_\mu :\widetilde J^{-1}(\mu) \to T^*\widetilde Q$ is the natural inclusion and $\rho_{\mbox{\tiny{$\mu$}}} :\widetilde J^{-1}(\mu)\to \widetilde J^{-1}(\mu)/F$ is the orbit projection.
  \item[$(iii)$] [The reduced dynamics]  The reduced nonholonomic vector field $X_{\emph\red}$ on $\M/G$ is tangent to the manifold  $\widetilde J^{-1}(\mu)/F$ for $\mu= c_i\mu^i$, with $c_i\in\R$ and $\mu^i\in \widetilde{\mathfrak{B}}^*_{\emph{\mbox{\tiny{HGS}}}}$, and its restriction to this leaf is a hamiltonian vector field for the 2-form $\omega^\BB_\mu$ and the hamiltonian function $H_\mu := (\iota_\mu^{\emph\red})^* H_{\emph\red}$, for $\iota^{\emph\red}_\mu:\widetilde J^{-1}(\mu)/F \to T^*\widetilde Q/F$ the inclusion. 
 \end{enumerate}
\end{theorem}

\begin{proof}
$(i)$ Since the nonholonomic system $(\M,\Omega_\subM|_\C, H_\subM)$ is $G$-invariant and $\widetilde{\bf B}$ is $F$-invariant, then the partially reduced system $(T^*\widetilde Q, \widetilde \Omega_\BB, \widetilde{H})$ is invariant by the $F$-action defined in \eqref{Eq:F-action}.  Moreover, the canonical momentum map $\widetilde J:T^*\widetilde Q \to  \mathfrak{f}^*$ and the basis $\widetilde{\mathfrak{B}}_{\mbox{\tiny{HGS}}}$ define the functions $\widetilde J_i$ which are partially reduced horizontal gauge momenta, see Lemma \ref{L:lemma1-1}, Prop.~\ref{Prop:tildeJ_i} and Def.~\ref{Def:RedHorGauge}.

$(ii)$ Following Prop.~\ref{P:LevelSets}, for each $\mu = c_i\mu^i$ (for $c_i\in\R$ and $\mu^i \in \widetilde{\mathfrak{B}}^*_{\mbox{\tiny{HGS}}}$),  $\widetilde J^{-1}(\mu)$ is a $F$-invariant manifold and, since the $F$-action is free and proper, the quotient space  $\widetilde J^{-1}(\mu)/F$ is a well defined manifold. 
Let us denote by $\widetilde \Omega_\mu^\BB$ the pull back of $\widetilde \Omega_\BB$ to $\widetilde J^{-1}(\mu)$, i.e., $ \widetilde \Omega_\mu^\BB := \iota_\mu^*  \widetilde \Omega_\BB$. Next, we will show that  $\widetilde \Omega_\mu^\BB$ is basic with respect to the bundle $\widetilde J^{-1}(\mu) \to \widetilde J^{-1}(\mu)/F$.    That is, as a consequence of Prop.~\ref{Prop:HGM-MomMap}, we will prove that, for each $\alpha\in \widetilde J^{-1}(\mu) \subset T^*\widetilde Q$, 
\begin{equation}\label{Eq:Proof:Ker=Torb}
\textup{Ker}\,( \widetilde \Omega_\mu^\BB {\mbox{\tiny{$(\alpha)$}}} ) = T_{\alpha}(\textup{Orb}_F(\alpha)),  
\end{equation}
where $\textup{Orb}_F(\alpha)$ is the orbit of the $F$-action at $\alpha$.  First, we claim that for all $\alpha\in \widetilde J^{-1}(\mu)$ and $\mu = c_i\mu^i$,  $T_{\alpha}\widetilde J^{-1}(\mu) = (T_{\alpha}\textup{Orb}_F(\alpha))^{\widetilde \Omega_\BB}$.
In fact, the flow $\phi^X_t$ of a vector field $X$ on $\widetilde J^{-1}(\mu)$,  satisfies that $\widetilde J(\phi^X_t(\alpha) ) = \mu$ for all $t$. Then, for $\eta_i\in  \widetilde{\mathfrak{B}}_{\mbox{\tiny{HGS}}}$ and using Prop. \ref{Prop:HGM-MomMap}, we have that
 $$
 \widetilde \Omega_\BB(X,(\eta_i)_{\mbox{\tiny{$T^*\!\widetilde Q$}}}(\alpha)) = - d\widetilde J_i(\alpha) \, X = - \left.\tfrac{d}{dt} \widetilde J_i(\phi_t^X(\alpha))\right|_{t=0} = - \left.\tfrac{d}{dt} \langle \widetilde J(\phi_t^X(\alpha)), \eta_i\rangle \right|_{t=0} = - \left. \tfrac{d}{dt} c_i \right|_{t=0} = 0. 
 $$
 Since $(\eta_i)_{\mbox{\tiny{$T^*\!\widetilde Q$}}}(\alpha)$ for $i=1,...,k$ form a basis of $T_{\alpha}\textup{Orb}_F(\alpha)$, then we obtain that $X(\alpha)\in (T_{\alpha}\textup{Orb}_F(\alpha))^{\widetilde \Omega_\BB}$.  Finally, $T_{\alpha}\widetilde J^{-1}(\mu) = (T_{\alpha}\textup{Orb}_F(\alpha))^{\widetilde \Omega_\BB}$ since both spaces have the same dimension.

Now we prove the identity \eqref{Eq:Proof:Ker=Torb}.  Let $X(\alpha) \in \textup{Ker}\, (\widetilde \Omega^\BB_\mu {\mbox{\tiny{$(\alpha)$}}})$, then $\widetilde \Omega^\BB_\mu (X(\alpha),Y(\alpha)) =0$ for all $Y(\alpha)\in T_{\alpha} \widetilde J^{-1}(\mu)=(T_{\alpha}\textup{Orb}_F(\alpha))^{\widetilde \Omega_\BB}$. 
Then $X(\alpha)\in [(T_{\alpha}\textup{Orb}_F(\alpha))^{\widetilde \Omega_\BB}]^{\widetilde \Omega_\BB} = T_{\alpha}\textup{Orb}_F(\alpha)$ (since the 2-form $\widetilde \Omega_\BB$ is nondegenerate). 
Conversely, for $\eta_i\in \widetilde{\mathfrak{B}}_{\mbox{\tiny{HGS}}}$, $\widetilde \Omega_\mu^\BB((\eta_i)_{\mbox{\tiny{$\widetilde J^{-1}(\mu)$}}}(\alpha) , X(\alpha)) =0$ for all
$X(\alpha)\in T_{\alpha}\widetilde J^{-1}(\mu)$ and then $(\eta_i)_{\mbox{\tiny{$\widetilde J^{-1}(\mu)$}}}(\alpha) \in \textup{Ker}\, ( \widetilde \Omega^\BB_\mu {\mbox{\tiny{$(\alpha)$}}})$ for all $i=1,...,k$.  

Therefore, by \eqref{Eq:Proof:Ker=Torb}, the 2-form $\widetilde \Omega_\mu^\BB$   descends to
 an almost symplectic 2-form $\omega_\mu^\BB$ on $\widetilde J^{-1}(\mu)/F$ such that $\rho_{\mbox{\tiny{$F$}}}^*\omega^\BB_\mu = \widetilde \Omega_\mu^\BB$.  

$(iii)$ Since $\widetilde J_i$ are conserved quantities for the partially reduced dynamics $\widetilde X_\nh$ on $T^*\widetilde Q$, then the flow $\widetilde{\phi}_t^{\nh}$ of the vector field $\widetilde X_\nh$ satisfies that, for $\alpha\in \widetilde J^{-1}(\mu)$,  $\widetilde{\phi}_t^{\nh}(\alpha) \in \widetilde J^{-1}(\mu)$ for all $t$. Therefore, by the $G$-invariance of the dynamics, we conclude that $\phi_t^\red(\rho_{\mbox{\tiny{$\mu$}}}(\alpha)) \in \widetilde J^{-1}(\mu)/F$ where $\phi_t^\red$ is the flow of the reduced nonholonomic dynamics $X_\red$ and  $\rho_{\mbox{\tiny{$\mu$}}} :\widetilde J^{-1}(\mu)\to \widetilde J^{-1}(\mu)/F$ the orbit projection. Denoting by $X^\mu_\red$ the restriction of $X_\red$ to the leaf $\widetilde J^{-1}(\mu)/F$,  we can see that ${\bf i}_{X_\red^\mu}\omega_\mu^\BB = d H_\mu$ as a direct consequence of item $(ii)$ and Prop.~\ref{P:B-Chaplygin} $(ii)$.

\end{proof}

\begin{remark} \label{R:G-invJ} The fact that we consider $G$-invariant horizontal gauge momenta $J_1,...,J_k$  and not only $G_\subW$-invariant, permits us to reduced the manifold $\widetilde J^{-1}(\mu)$ by the action of the Lie group $F$ (without taking into account any ``isotropy group'').  
Under this assumption, we may denote by $\bar{J}_i$ the functions on $T^*\widetilde Q/F$ such that $\rho^*_{\mbox{\tiny{$F$}}} \bar{J}_i = \widetilde{J}_i$ for $\rho_{\mbox{\tiny{$F$}}}:T^*\widetilde Q\to T^*\widetilde Q/F$ the corresponding  orbit projections (or equivalently $\rho^*\bar{J}_i = {J}_i$). Therefore,  $\widetilde J^{-1}(\mu)/F$ coincides with the common level sets of the reduced horizontal gauge momenta $\bar J_i$, i.e., $\widetilde J^{-1}(\mu)/F \simeq \cap_i \bar{J}_i^{-1}(c_i)$.  
\end{remark}

\section{The identification of $(J^{-1}(\mu)/F,\omega^\BB_\mu)$ with the canonical symplectic manifold}\label{S:Identifications}

In the hamiltonian framework, when working on a canonical symplectic manifold $(T^*Q, \Omega_\subQ)$ (and when $G=G_\mu$)  we have the identification of the Marsden-Weinstein reduced symplectic manifolds with the cotangent manifold $T^*(Q/G)$ and its canonical symplectic form plus a magnetic term that depends on a chosen connection, see e.g. \cite{AbrMarsden,Marsdenetal}.   In this section we show an analogous identification but carrying on the information of the nonholonomic character of the system. That is, we take into account that the nonholonomic dynamics takes place on the almost symplectic manifold $(T^*\widetilde Q, \widetilde \Omega_\BB)$ and that the 2-form $\widetilde{\bf B}$ also depends on a chosen connection. Hence we obtain an  identification of the ``Marsden-Weinstein'' reduced spaces $(\widetilde J^{-1}(\mu)/F, \omega_\mu^\BB)$ with the cotangent manifold $T^*(\widetilde Q/F)$ and its canonical symplectic form modified by a  ``magnetic''  term, i.e., a 2-form that, in this case, does not come from a 2-form on $\widetilde Q/F$ (as in the hamiltonian case) and it depends only on the 2-form $\widetilde{\mathcal{B}}$.
This extra term carries the nonholonomic character of the reduced system since, contrary to hamiltonian systems, its differential can be different from zero. 
Moreover, in Examples \ref{Ex:NHParticle}, \ref{Ex:Snakeboard} and \ref{Ex:Solids} we will see that $\mathcal{B}=0$ and then the manifolds  $(\widetilde J^{-1}(\mu)/F, \omega_\mu^\BB)$ are diffeomorphic to the canonical symplectic manifold $(T^*(\widetilde{Q}/F), \Omega_{\mbox{\tiny{$\widetilde{Q}/F$}}})$ (showing a genuine hamiltonization). 

\medskip

Next, we consider, as usual, a nonholonomic system $(\M, \Omega_\subM|_\C, H_\subM)$ with a $G$-symmetry admitting $k=\textup{rank}(S)$ horizontal gauge momenta and with the vertical symmetry condition. Then the partially reduced nonholonomic system $(T^*\widetilde Q, \widetilde \Omega_\BB, \widetilde H)$ admits a symmetry given by the action of the Lie group $F$ with $k$ partially reduced horizontal gauge momenta (recall that the dimension of the Lie algebra $\mathfrak{f}$ is also $k$). 

\subsection{Identification at the zero-level}

Following \cite{Marsdenetal}, we consider the zero level set of the canonical momentum map $\widetilde J: T^*\widetilde Q\to \mathfrak{f}^*$ and the map $\widetilde{\varphi}_0:\widetilde J^{-1}(0)\to T^*\overline{Q}$, for $\overline{Q}:=\widetilde{Q}/F$  given, at each $\alpha_x\in \widetilde J^{-1}(0) \subset T^*\widetilde Q$, by 
$$
\langle \widetilde{\varphi}_0(\alpha_x), T\rho_{\mbox{\tiny{$\overline Q$}}} (v_x) \rangle = \langle \alpha_x, v_x\rangle,
$$ 
for $v_x\in T_x\widetilde Q$ and $\rho_{\mbox{\tiny{$\overline Q$}}}: \widetilde Q\to \overline{Q}$ the orbit projection.  
Since the map $\widetilde{\varphi}_0$ is $F$-invariant, it is shown also in \cite{Marsdenetal} that there is a well defined diffeomorphism 
$$
\varphi_0:  \widetilde J^{-1}(0)/F \to T^*\overline{Q},
$$
so that $\varphi_0 \circ \rho_0 = \widetilde{\varphi}_0$ for $\rho_0:  \widetilde J^{-1}(0) \to  \widetilde J^{-1}(0)/F$ the canonical projection.
Next, we show that this map is, in fact, the diffeomorhism that links the 2-form $\omega_0^\BB$ on $\widetilde J^{-1}(0)/F$ from Theorem \ref{T:MW} (at $\mu=0$) with the canonical 2-form $\Omega_{\mbox{\tiny{$\overline{Q}$}}}$ on $T^*\overline{Q}$.   

Recall, from Prop.~\ref{Prop:B_G-inv} and Prop.~\ref{P:B-Chaplygin}, that the 2-form $\widetilde{\bf B}$ can be written as $\widetilde{\bf B} = \widetilde B_1 + \widetilde{\mathcal{B}}$ where $\widetilde{\mathcal{B}}$ is a basic form with respect to the bundle $\rho_{\mbox{\tiny{$F$}}}:T^*\widetilde{Q} \to T^*\widetilde{Q}/F$ and hence we denote by $\overline{\mathcal{B}}$ the 2-form on $T^*\widetilde{Q}/F$ such that $\rho_{\mbox{\tiny{$F$}}}^*\overline{\mathcal{B}}=\widetilde{\mathcal{B}}$.   

\begin{proposition} \label{Prop:0Level}
 The diffeomorphism $\varphi_0 : \widetilde J^{-1}(0)/F \to T^*\overline{Q}$ satisfies that 
 $$
 \varphi_0^*\, \Omega_{\mbox{\tiny{$\overline{Q}$}}}  = \omega_0^\BB -  \overline{\mathcal{B}}_0,
 $$
 where $\Omega_{\mbox{\tiny{$\overline{Q}$}}}$ is the canonical 2-form on $T^*\overline{Q}$ and $\overline{\mathcal{B}}_0 := (\iota_0^{\emph\red})^*\overline{\mathcal B}$,  for $\iota_0^{\emph\red}: \widetilde J^{-1}(0)/F \to T^*\widetilde Q/F$ the natural inclusion. In particular, if $\textup{dim}(\overline{Q})=1$, then $\varphi_0^*\, \Omega_{\mbox{\tiny{$\overline{Q}$}}}  = \omega_0^\BB$. 
\end{proposition}

\begin{proof}
 On the one hand, it was shown in \cite{Marsdenetal} that the diffeomorphism $\varphi_0 : \widetilde J^{-1}(0)/F \to T^*\overline{Q}$ satisfies that 
 $ \varphi_0^* \, \Omega_{\mbox{\tiny{$\overline{Q}$}}} = \omega_0$, where $\omega_0$ is the symplectic form on $\widetilde J^{-1}(0)/F$ such that 
 $\rho_0^* \, \omega_0 = \iota_0^*\Omega_{\mbox{\tiny{$\widetilde Q$}}}$, for $\iota_0: \widetilde J^{-1}(0)\to T^*\widetilde Q$ the natural inclusion. 
 On the other hand, Theorem \ref{T:MW} at the zero-level, implies that 
 $\rho_0^*\omega_0^\BB = \iota_0^*\widetilde \Omega_\BB$.  From the expression of $\widetilde{B}_1$ in \eqref{Eq:tildeB1}, we have that  $\iota_0^*(B_{\mbox{\tiny{$\!\langle\! J\mathcal{K}\!\rangle$}}} - \widetilde{B}_1)= - \iota_0^*(\widetilde{J}_i d\widetilde{\mathcal{Y}}^i) = 0$. Therefore, we obtain that 
 $$
 \rho_0^*\, \omega_0^\BB = \iota_0^*(\Omega_{\mbox{\tiny{$\widetilde Q$}}} - B_{\mbox{\tiny{$\!\langle\! J\mathcal{K}\!\rangle$}}} +\widetilde{\bf B})  =\iota_0^*\, \Omega_{\mbox{\tiny{$\widetilde Q$}}} + \iota^*_0\, \widetilde{\mathcal{B}} = \rho_0^*\, \omega_0 + \iota_0^* \, \rho_{\mbox{\tiny{$F$}}}^*\, \overline{\mathcal{B}} = \rho_0^* \omega_0 + \rho_0^* \, (\iota_0^{\red})^* \, \overline{\mathcal{B}},
 $$
 where in the last equality we used that $\iota_0^\red \circ \rho_0 = \rho_{\mbox{\tiny{$F$}}} \circ \iota_0$ for $\rho_{\mbox{\tiny{$F$}}}: T^*\widetilde Q\to T^*\widetilde Q/F$ the orbit projection.  Then $\omega_0^\BB = \omega_0 + (\iota_0^{\red})^* \, \overline{\mathcal{B}}$ which implies that $\omega_0^\BB = \varphi_0^* \, \Omega_{\mbox{\tiny{$\overline{Q}$}}} + \overline{\mathcal{B}}_0$. 
 
\end{proof}

\subsection{Identification at the $\mu$-level and the Shift-trick}

Now, using the {\it Shift-trick} as in \cite{AbrMarsden,Marsdenetal}, we show that each (connected component of the) almost symplectic manifold $(\widetilde J^{-1}(\mu)/F, \omega_\mu^\BB)$ obtained in Theorem~\ref{T:MW}, is diffeomorphic to $T^*\overline{Q}$ with its canonical 2-form $\Omega_{\mbox{\tiny{$\overline{Q}$}}}$  properly modified by a ``magnetic'' term.

As usual, we denote by $\widetilde{\mathfrak{B}}_{\mbox{\tiny{HGS}}}=\{\eta_1,...,\eta_k\}$ a global basis   of equivariant $\mathfrak{f}$-valued functions on $\widetilde Q$  of (partially reduced) horizontal gauge symmetries 
and  $\widetilde{\mathfrak{B}}^*_{\mbox{\tiny{HGS}}} = \{\mu^1,...,\mu^k\}$ the dual basis of $\mathfrak{f}^*$-valued functions given in \eqref{Eq:RedDualBasisHGS}.
Recall that $\widetilde A$ is the induced connection on $T\widetilde Q$ (see Lemma \ref{L:lemma1-1}) and observe that the 2-form $\widetilde{B}_1$ in \eqref{Eq:tildeB1} is written with respect to this connection: $\widetilde A= \eta_i\otimes \widetilde{Y}^i$ and $\widetilde B_1 =  B_{\mbox{\tiny{$\!\langle\! J\mathcal{K}\!\rangle$}}} + \widetilde J_i d\widetilde{\mathcal Y}^i$ where $\widetilde{\mathcal{Y}}^i= \tau_{\widetilde\subQ}^*\widetilde{Y}^i$  for $\tau_{\widetilde\subQ}: T^*\widetilde Q \to \widetilde Q$ is the canonical projection. 

 Next we proceed to define the $\textup{Shift}$-map that, on $\mu^i\in \widetilde{\mathfrak{B}}^*_{\mbox{\tiny{HGS}}}$, coincides with the one defined in \cite{AbrMarsden, Marsdenetal}.
More precisely, for $\mathfrak{f}^*$-valued functions $\mu = c_i\mu^i$, where $c_i\in \R$ and $\mu^i \in \widetilde{\mathfrak{B}}^*_{\mbox{\tiny{HGS}}}$, we define the diffeomorphism 
\begin{equation}
 \begin{split}
\textup{Shift}_\mu : T^*\widetilde Q & \to T^*\widetilde Q\\
  \alpha & \mapsto \alpha - \alpha_\mu,
 \end{split}
\end{equation}
where $\alpha_\mu = \langle \mu, \widetilde{A}\rangle$ for $\langle \cdot, \cdot \rangle$ the natural pairing between the $\mathfrak{f}^*$-valued function $\mu$ and the $\mathfrak{f}$-valued 1-form $\widetilde A$, i.e., for $x\in \widetilde Q$,  $\alpha_\mu(x) = \langle \mu(x), \widetilde A_x\rangle \in T^*_x\widetilde Q$. 


\begin{lemma}\label{L:shift}
 If  $\mu=c_i\mu^i$, for $c_i\in\R$ and $\mu^i\in \widetilde{\mathfrak{B}}^*_{\emph{\mbox{\tiny{HGS}}}}$, then 
 \begin{enumerate}
  \item[$(i)$] $\textup{Shift}_\mu^*\, \Omega_{\mbox{\tiny{$\widetilde Q$}}} = \Omega_{\mbox{\tiny{$\widetilde Q$}}} + \tau_{\widetilde\subQ}^*\, c_i d\widetilde Y^i$.
  \item[$(ii)$] The restricted map $\textup{shift}_\mu :=\textup{Shift}_\mu|_{\widetilde J^{-1}(\mu)} : \widetilde J^{-1}(\mu) \to \widetilde J^{-1}(0)$ is a well defined equivariant diffeo\-mor\-phism and hence there is a well defined diffeomorphism $\overline{\textup{shift}}_\mu : \widetilde J^{-1}(\mu)/F \to \widetilde J^{-1}(0)/F$ so that the following diagram commutes
 \begin{equation}\label{Diag:shift}
  \xymatrix{ T^*\widetilde Q  \ar[d]^{\textup{Shift}_\mu}  && {\widetilde{J}}^{-1}(\mu) \ar[ll]_{\iota_\mu} \ar[rr]^{\rho_\mu} \ar[d]^{\textup{shift}_\mu} && {\widetilde{J}}^{-1}(\mu)/F \ar[d]^{\overline{\textup{shift}}_\mu} \\
 T^*\widetilde Q  && {\widetilde{J}}^{-1}(0)  \ar[ll]_{\iota_0} \ar[rr]^{\rho_0}  && {\widetilde{J}}^{-1}(0)/F  
  }
  \end{equation}

 \end{enumerate}

\end{lemma}

\begin{proof}
$(i)$ Note that if  $\mu=c_i\mu^i$, for $\mu^i\in \widetilde{\mathfrak{B}}^*_{\mbox{\tiny{HGS}}}$, then by \eqref{Eq:TildeA},  $\langle \mu, \widetilde{A}\rangle = c_i \widetilde Y^i$ and hence $d\langle \mu, \widetilde{A}\rangle = c_i d\widetilde Y^i$.  Moreover, following \cite{Marsdenetal}, we can also prove that $\textup{Shift}_\mu^*\Theta_{\mbox{\tiny{${\widetilde Q}$}}} = \Theta_{\mbox{\tiny{${\widetilde Q}$}}} -\tau_{\widetilde\subQ}^* \langle \mu, \widetilde{A}\rangle$ 
and conclude that $\textup{Shift}_\mu^*\, \Omega_{\mbox{\tiny{$\widetilde Q$}}} = \Omega_{\mbox{\tiny{$\widetilde Q$}}} + \tau_{\widetilde\subQ}^* \, c_i d\widetilde Y^i$.

$(ii)$ It is straightforward to check that $\textup{shift}_\mu$ is a diffeomorphism. To see the equivariance,  recall that, for $h\in F$, $\widetilde \Psi_h:\widetilde{Q}\to \widetilde{Q}$ denotes the $F$-action on $\widetilde Q$  and $\mu(x)$ denotes the evaluation of the $\mathfrak{f}^*$-valued function $\mu$ at $x\in \widetilde{Q}$. On the one hand, due to the $F$-invariance of the horizontal space $\widetilde{H\!{\mbox{\tiny{or}}}}$, the connection $\widetilde A$ is $Ad$-equivariant: for $x\in \widetilde Q$, $y=\Psi_{h^{-1}}(x)$, and $X$ a vector field on $\widetilde Q$.  $\widetilde A_x(T\widetilde \Psi_h(X(y))) = Ad_h (\widetilde A_y(X(y)))$. On the other hand, since the (partially reduced) horizontal gauge momenta $\widetilde J_i$ are $F$-invariant functions, then the associated horizontal gauge symmetries $\eta_i$ (seen as $\mathfrak{f}$-valued functions) are $Ad$-equivariant and the corresponding $\mathfrak{f}^*$-valued functions $\mu^i$ are $Ad^*$-equivariant, i.e.,  $Ad_h(\eta_y) = \eta_x$ and $Ad_h^*\mu(x) = \mu(y)$.  

For $\alpha_x\in T^*_x\widetilde Q$ and $h\in F$, we get that
\begin{equation*}
	\begin{split}
		T^*\widetilde\Psi_h(\textup{Shift}_\mu(\alpha_x)) & = T^*\widetilde\Psi_h(\alpha_x) - T^*\widetilde\Psi_h(\langle \mu(x),\widetilde{A}_x \rangle)  = T^*\widetilde\Psi_h(\alpha_x) - \langle \mu(x),\widetilde{A}_x\circ T{\widetilde \Psi_h}) \rangle  \\
		& = T^*\widetilde\Psi_h(\alpha_x) - \langle \mu(x),Ad_h(\widetilde{A}_y) \rangle  = T^*\widetilde\Psi_h(\alpha_x) - \langle Ad_{h}^*(\mu(x)),\widetilde{A}_{y} \rangle \\
		& = T^*\widetilde\Psi_h(\alpha_x) - \langle \mu(y),\widetilde{A}_{y} \rangle =\textup{Shift}_\mu(T^*\widetilde\Psi_h(\alpha_x)).
	\end{split}
\end{equation*}

\end{proof}

Next, for each $\mathfrak{f}^*$-valued function $\mu=c_i\mu^i$ with $c_i\in \R$ and $\mu^i\in \widetilde{\mathfrak{B}}^*_{\mbox{\tiny{HGS}}}$ we consider the map 
\begin{equation}\label{Def:varphi_mu} 
\varphi_\mu := \varphi_0 \circ \overline{\textup{shift}}_\mu: \widetilde{J}^{-1}(\mu)/F \to T^*\overline{Q}, 
\end{equation}
 which, by construction, is a diffeomorphism. 

\begin{theorem} \label{T:mu-identification}
 Consider a nonholonomic system $(\M, \Omega_\subM|_\C,H_{\subM})$ with a $G$-symmetry satisfying the di\-men\-sion assumption such that it admits $k=\textup{rank}(S)$ $G$-invariant horizontal gauge momenta.  More\-o\-ver, we assume that the system verifies the vertical symmetry condition.  
The reduction of the partially reduced system $(T^*\widetilde Q, \widetilde \Omega_\BB, \widetilde H)$ given in Theorem~\ref{T:MW}, induces, for each $\mu=c_i\mu^i$ (with $\mu^i\in  \widetilde{\mathfrak{B}}^*_{\emph{\mbox{\tiny{HGS}}}}$),  the almost symplectic ma\-ni\-fold $(\widetilde J^{-1}(\mu)/F, \omega_\mu^\BB)$ for which the diffeomorphism $\varphi_\mu : \widetilde J^{-1}(\mu)/F \to T^*\overline{Q}$ satisfies that 
 $$
 \varphi_\mu^*\, \Omega_{\mbox{\tiny{$\overline{Q}$}}} = \omega_\mu^\BB - \overline{\mathcal{B}}_\mu,
 $$
 where $\overline{\mathcal{B}}_\mu := (\iota_\mu^{\emph\red})^*\overline{\mathcal B}$ for $\iota_\mu^{\emph\red} : \widetilde J^{-1}(\mu)/F \to T^*\widetilde Q/F$ the natural inclusion.
In particular, if $\textup{dim}(\overline{Q}) =1$, then $\varphi_\mu^*\Omega_{\mbox{\tiny{$\overline{Q}$}}} = \omega^{\BB}_{\mu}$.
\end{theorem}

\begin{proof}
 First, consider the following two commutative diagrams
 \begin{equation}\label{Diag:iotas}
 \xymatrix{{\widetilde{J}}^{-1}(0) \ar[d]_{\rho_0} \ar[r]^{\iota_0} & T^*\widetilde Q \ar[d]_{\rho_{\mbox{\tiny{$F$}}} }\\
 {\widetilde{J}}^{-1}(0)/F \ar[r]^{\iota_0^{\red}} & T^*\widetilde Q /F } 
 \qquad  \qquad
  \xymatrix{{\widetilde{J}}^{-1}(\mu) \ar[d]_{\rho_\mu} \ar[r]^{\iota_\mu} & T^*\widetilde Q \ar[d]_{\rho_{\mbox{\tiny{$F$}}} }\\
{\widetilde{J}}^{-1}(\mu)/F \ar[r]^{\iota_\mu^{\red}} & T^*\widetilde Q /F }
 \end{equation}
Since $\rho_{\mbox{\tiny{$F$}}}^*\, \overline{\mathcal{B}} = \widetilde{\mathcal{B}}$ and, as a consequence of Theorem \ref{T:MW}, we see that $\varphi_\mu^*\, \Omega_{\mbox{\tiny{$\overline{Q}$}}} = \omega_\mu^\BB - \overline{\mathcal{B}}_\mu$ if and only if
 \begin{equation}\label{Eq:Proof:muIdent}
 \rho_\mu^* \circ \varphi_\mu^* \, \Omega_{\mbox{\tiny{$\overline{Q}$}}} = \iota_\mu^*(\Omega_{\mbox{\tiny{$\widetilde{Q}$}}} - B_{\mbox{\tiny{$\!\langle\! J\mathcal{K}\!\rangle$}}} +\widetilde{\bf B} ) - \iota_\mu^*\widetilde{\mathcal{B}}.
 \end{equation}
 Next, we will prove \eqref{Eq:Proof:muIdent}. Using the definition of $\varphi_\mu$ in \eqref{Def:varphi_mu} and by \eqref{Diag:shift} we have that
 \begin{equation*}
   \rho_\mu^* \circ \varphi_\mu^*\, \Omega_{\mbox{\tiny{$\overline{Q}$}}}  = \rho_\mu^* \circ \overline{\textup{shift}}_\mu^* \circ \varphi_0^* \, \Omega_{\mbox{\tiny{$\overline{Q}$}}} = \textup{shift}_\mu^* \circ \rho_0^* \circ \varphi_0^* \,  \Omega_{\mbox{\tiny{$\overline{Q}$}}} = \textup{shift}_\mu^* \circ \rho_0^* \, (\omega_0^\BB - \overline{\mathcal{B}}_0),
 \end{equation*}
where in the last equality we used Prop.~\ref{Prop:0Level}. Moreover, since $\rho_0^*(\omega_0^\BB - \overline{\mathcal{B}}_0) = \iota_0^*\, \Omega_{\mbox{\tiny{$\widetilde{Q}$}}}$ (see the proof of Prop.~\ref{Prop:0Level}), and using \eqref{Diag:iotas} and Lemma \ref{L:shift}, we conclude that
\begin{equation*}
   \rho_\mu^* \circ \varphi_\mu^* \, \Omega_{\mbox{\tiny{$\overline{Q}$}}}  = \textup{shift}_\mu^* \circ \iota_0^* \, \Omega_{\mbox{\tiny{$\widetilde{Q}$}}} = \iota_\mu^*\circ \textup{Shift}_\mu^* \, \Omega_{\mbox{\tiny{$\widetilde{Q}$}}}= \iota_\mu^* (\Omega_{\mbox{\tiny{$\widetilde{Q}$}}} +\tau_{\widetilde\subQ}^*\, c_i d\widetilde{Y}^i)= \iota_\mu^* (\Omega_{\mbox{\tiny{$\widetilde{Q}$}}} + \widetilde J_i d\widetilde{\mathcal Y}^i),
\end{equation*}
where in the last equality we also used Prop.\ref{P:LevelSets} and the fact that $\widetilde{\mathcal{Y}}^i=\tau_{\widetilde\subQ}^*\widetilde Y^i$. 
Finally, recalling the expression of ${\widetilde B}_1$ in \eqref{Eq:tildeB1}, we see that $\Omega_{\mbox{\tiny{$\widetilde{Q}$}}} + \widetilde J_i d\widetilde{\mathcal Y}^i = \Omega_{\mbox{\tiny{$\widetilde{Q}$}}} - B_{\mbox{\tiny{$\!\langle\! J\mathcal{K}\!\rangle$}}} +\widetilde B_1$ and, since $\widetilde B_1 = \widetilde{\bf B} - \widetilde{\mathcal{B}}$, we arrive to the desired result \eqref{Eq:Proof:muIdent}.  
\end{proof}

\begin{remark} \label{R:Indentification}
Theorem \ref{T:mu-identification} identifies each almost symplectic manifold $(\widetilde J^{-1}(\mu)/F, \omega_ \mu^\BB)$ with $(T^*\overline{Q}, \Omega_{\mbox{\tiny{$\overline{Q}$}}} + \widehat{\mathcal{B}}_\mu)$ where $\widehat{\mathcal{B}}_\mu := (\varphi_\mu^{-1})^*\overline{\mathcal{B}}_\mu$.  Observe that $\widehat{\mathcal{B}}_\mu$ is not a magnetic term in the strict sense since it might be non closed and is not coming from a 2-form defined on $\overline{Q}$. Moreover, it has no connection with the magnetic term that appears in hamiltonian systems. 
\end{remark}

\section{Relation with the nonholonomic bracket} \label{S:NHBracket}

In this section we will work with the {\it nonholonomic bracket} $\{\cdot, \cdot \}_\nh$ on $\M$ defined by the system $(\M, \Omega_\subM|_\C, H_\subM)$. 
First, we recall from \cite{BGN} how the gauge transformation by the 2-form ${\bf B}$ of the nonholonomic bracket defines a new almost Poisson bracket $\{\cdot, \cdot\}_\BB$ on $\M$. Afterwards, we will see how the 2-step reduction developed in Section \ref{S:Reduction} translates into a 2-step reduction where at each level we obtain an almost Poisson bracket showing that the almost symplectic foliation defined in Theorem \ref{T:MW} coincides with the union of almost symplectic leaves associated to the reduced bracket $\{\cdot, \cdot\}_\red^\BB$ on $\M/G$. 

\subsection{The nonholonomic bracket and reduction}\label{Ss:Gauge-Leaves}

Consider a nonholonomic system given by the triple $(\M, \Omega_\subM|_\C, H_\subM)$ as in Section \ref{S:Prelim}. The fact that $\Omega_\subM|_\C$ is a nondegenerate 2-section \cite{BS93} not only defines uniquely the nonholonomic vector field but also it induces an almost Poisson bracket $\{\cdot, \cdot \}_\nh$ on functions on $\M$, given for each $f\in C^\infty(\M)$, by 
\begin{equation}\label{Eq:NH-Bracket}
 \{\cdot, f\}_\nh =X_f \qquad \mbox{if and only if} \qquad {\bf i}_{X_f} \Omega_\subM |_\C = (df) |_\C.
\end{equation}
The bracket $\{\cdot, \cdot \}_\nh$ is called {\it the nonholonomic bracket} \cite{IdLMdD,Marle,SchaftMaschke1994} and it describes the nonholonomic dynamics since 
$$
X_\nh = \{ \cdot, H_\subM \}_\nh. 
$$
Recall that, on the one hand, an almost Poisson bracket is a bilinear, skew-symmetric bracket that satisfies Leibniz identity but not necessarily the Jacobi identity.  In fact, the {\it characteristic distribution} of the nonholonomic bracket $\{\cdot, \cdot \}_\nh$ --the distribution generated by the hamiltonian vector fields $X_f$-- is the nonintegrable distribution $\C$, induced by the constraints and defined in \eqref{Eq:C},  and hence the Jacobi identity is not satisfied.
On the other hand, a Poisson bracket satisfies the Jacobi identity and, as a consequence, it has an integrable characteristic distribution inducing a symplectic foliation. In between, there is a class of almost Poisson brackets that have an integrable characteristic distribution but the Jacobi identity is still not satisfied.  More precisely, an almost Poisson bracket $\{\cdot, \cdot\}$ on $\M$ is {\it twisted Poisson}   \cite{KlimStro-2002,SeveraWeinstein} if there exists a closed 3-form $\Phi$ on $\M$ such that 
\begin{equation}\label{Def:twisted}
\{f,\{g,h\}\} + \{g,\{h,f\}\} + \{h,\{f,g\}\} = \Phi(X_f, X_g,X_h), \quad \mbox{ for } f,g,h \in C^\infty(\M). 
\end{equation}
A twisted Poisson bracket admits an almost symplectic foliation. The role of (regular) twisted Poisson brackets in nonholonomic mechanics was studied in \cite{BJac,BGN,BY} where it was observed that the reduction by symmetries of the nonholonomic bracket might become twisted Poisson. 

If the nonholonomic system $(\M, \Omega_\subM|_\C, H_\subM)$ has a symmetry given by the action of a Lie group $G$, then the nonholonomic bracket is $G$-invariant as well, and thus it can be reduced to an almost Poisson bracket $\{\cdot, \cdot \}_\red$ on $\M/G$ so that, for $\bar f, \bar g\in C^\infty (\M/G)$, 
\begin{equation}\label{Eq:RedBracket}
\{\bar f, \bar g\}_\red (\rho (m)) = \{\rho^* \bar f, \rho^*\bar g\}_\nh (m),
\end{equation}
for $m\in \M$ and $\rho: \M \to \M/G$ the orbit projection.  The reduced bracket $\{\cdot, \cdot \}_\red$ on $\M/G$ is responsible of the reduced dynamics: 
$$
X_\red = \{\cdot, H_\red \}_\red, 
$$
where $H_\red : \M/G \to \R$ is, as usual, the reduced hamiltonian, i.e., $\rho^*H_\red = H_\subM$. 

In this section, we denote a nonholonomic system by the triple $(\M, \{\cdot, \cdot\}_\nh, H_\subM)$ or by \linebreak $(\M/G, \{\cdot, \cdot\}_\red, H_\red)$ to refer to the reduced system. 

Following Section \ref{Ss:VerticalSymmetries}, we assume that we can choose a vertical complement of the constraints that satisfies the vertical symmetry condition,  then the reduction by $G_\subW$ of the nonholonomic system $(\M, \{\cdot, \cdot \}_\nh,H_\subM)$ gives the partially reduced nonholonomic system $(\M/G_\subW, \{\cdot, \cdot\}_{\widetilde\nh},\widetilde H)$ with the almost Poisson bracket $\{\cdot, \cdot\}_{\widetilde\nh}$ obtained as in \eqref{Eq:RedBracket} but with respect to the orbit projection $\rho_{\mbox{\tiny{$G_{\!W}$}}}: \M\to \M/G_\subW$. It is straightforward to see that the bracket $\{\cdot, \cdot\}_{\widetilde\nh}$ is nondegenerate and hence the system $(\M/G_\subW, \{\cdot, \cdot\}_{\widetilde\nh}, \widetilde H)$ is just the Chaplygin system $(T^*\widetilde Q, \widetilde\Omega, \widetilde H)$ from \eqref{Eq:DynChaplygin}.   Finally, the reduction by the Lie group $F=G/G_\subW$ of  $(\M/G_\subW, \{\cdot, \cdot\}_{\widetilde\nh}, \widetilde H)$ gives the reduced nonholonomic system  $(\M/G, \{\cdot, \cdot\}_\red, H_\red)$.

\subsection{Gauge transformation of the nonholonomic bracket and reduction}

Consider a nonholonomic system $(\M, \Omega_\subM|_\C, H_\subM)$ with a $G$-symmetry admitting $k=\textup{rank}(S)$ horizontal gauge momenta $\{J_1,...,J_k\}$. 
Next, we study the gauge transformation by the 2-form ${\bf B}$ of the nonholonomic bracket $\{\cdot, \cdot\}_\nh$ and its reduction process.

More precisely, consider the 2-form ${\bf B}$ on $\M$ defined in \eqref{Def:B}. Observe that, since $\Omega_\subM|_\C$ is nondegenerate and ${\bf B}$ is semi-basic with respect to the bundle $\tau_\subM:\M\to Q$, then $(\Omega_\subM + {\bf B})|_\C$ is a nondegenerate 2-section and we can define the new bracket $\{\cdot, \cdot\}_\BB$ on functions on $\M$, given at each $f\in C^\infty(\M)$, by
\begin{equation} \label{Eq:GaugedBracket}
X_f = \{\cdot, f\}_\BB  \quad \mbox{if and only if} \quad  {\bf i}_{X_f}(\Omega_\subM + {\bf B})|_\C = (df)|_\C,
\end{equation}
(c.f. \eqref{Eq:NH-Bracket}).  Therefore the 2-form ${\bf B}$ defines a {\it gauge transformation} of the nonholonomic bracket $\{\cdot, \cdot\}_\nh$ generating the {\it gauge related} bracket $\{\cdot, \cdot\}_\BB$ on $\M$ (see \cite{BGN,GN2008,SeveraWeinstein} for more details). 
Since the 2-form ${\bf B}$ satisfies the dynamical condition ${\bf i}_{X_\nh}{\bf B} = 0$ (see \eqref{Eq:DynamicCondition}), then the bracket $\{\cdot, \cdot\}_\BB$ still describes the nonholonomic dynamics: $X_\nh = \{\cdot, H_\subM\}_\BB$ and we say that ${\bf B}$ defines a {\it dynamical gauge transformation} \cite{BGN}.  
%
From \eqref{Eq:NH-Bracket} and \eqref{Eq:GaugedBracket} we see that the brackets $\{\cdot,\cdot\}_\nh$ and $\{\cdot, \cdot\}_\BB$ share the characteristic distribution $\C$ (hence $\{\cdot,\cdot\}_\BB$ is an almost Poisson bracket as well).

Due to  Prop.~\ref{Prop:B_G-inv} the 2-form ${\bf B}$ is $G$-invariant, then the (dynamically) gauge related bracket $\{\cdot,\cdot\}_\BB$ is also $G$-invariant and it descends to an almost Poisson bracket $\{\cdot,\cdot\}_\red^\BB$ on $\M/G$, such that, at each  $\bar f, \bar g\in C^\infty (\M/G)$,  and for  $m\in \M$
\begin{equation} \label{Eq:BRedBracket}
\{\bar f, \bar g\}_\red^\BB (\rho (m)) = \{\rho^* \bar f, \rho^*\bar g\}_\BB (m).
\end{equation}

Therefore the nonholonomic system can be equivalently determined by the triple $(\M, \{\cdot, \cdot\}_\BB, H_\subM)$ and by $(\M/G, \{\cdot, \cdot\}_\red^\BB, H_\subM)$ to refer to the reduced system. 

Note that the bracket $\{\cdot, \cdot\}_{\BB}$ is, in particular, $G_\subW$-invariant as well. Following \eqref{Eq:BRedBracket}, but using the orbit projection $\rho_{\mbox{\tiny{$G_{\!W}$}}} : \M\to T^*\widetilde Q$, the almost Poisson bracket $\{\cdot, \cdot\}_{\BB}$ descends to an almost Poisson bracket $\{\cdot , \cdot \}_{\widetilde\BB}$ on $T^*\widetilde Q$.

\begin{proposition}\label{Prop:BracketChap}
 The almost Poisson bracket $\{\cdot , \cdot \}_{\widetilde\BB}$ on $T^*\widetilde Q$ is nondegenerate and defined by the 2-form $\widetilde \Omega_\BB=\widetilde \Omega + \widetilde{\bf B}$, i.e., for $f\in C^\infty(T^*\widetilde Q)$,  
 $$
 X_f = \{\cdot,f\}_{\widetilde\BB} \quad \mbox{ if and only if } \quad {\bf i}_{X_f}\widetilde \Omega_\BB= df. 
 $$
 
\end{proposition}

\begin{proof}
Since the brackets $\{\cdot, \cdot\}_\nh$ and $\{\cdot, \cdot\}_\BB$ are gauge related by the 2-form ${\bf B}$, which is basic  with respect to the bundle $\rho_{\mbox{\tiny{$G_{\!W}$}}}:\M\to T^*\widetilde Q$, then by \cite[Prop.4.8]{BJac}, the (partially reduced) brackets $\{\cdot, \cdot\}_{\widetilde\nh}$ and $\{\cdot, \cdot\}_{\widetilde\BB}$ on $T^*\widetilde Q$ are gauge related by the 2-form ${\widetilde{\bf B}}$.   Therefore, since $\widetilde \Omega$ is the 2-form associated to $\{\cdot, \cdot\}_{\widetilde\nh}$, then $\widetilde \Omega + \widetilde{\bf B}$ is the corresponding 2-form associated to $\{\cdot, \cdot\}_{\widetilde{\BB}}$. 
\end{proof}

From Prop.~\ref{Prop:BracketChap} we conclude that the brackets $\{\cdot, \cdot\}_{\widetilde\nh}$ and $\{\cdot, \cdot\}_{\widetilde\BB}$ are (dynamically) gauge related and that the triple $(T^*\widetilde Q, \{\cdot, \cdot\}_{\widetilde\BB}, \widetilde H)$ is just equivalent to the triple $(T^*\widetilde Q, \widetilde\Omega_\BB, \widetilde H)$.  
Moreover, from Theorem~\ref{T:MW}$(i)$ we observe that the nonholonomic system $(T^*\widetilde Q, \{\cdot , \cdot \}_{\widetilde\BB}, \widetilde H)$ has a symmetry given by the action of the Lie group $F=G/G_\subW$ and, using the orbit projection $\rho_{\mbox{\tiny{$F$}}}:T^*\widetilde Q\to T^*\widetilde Q/F$,  we can define the reduced bracket  $\{\cdot , \cdot \}_\red^{\BB}$ on $T^*\widetilde Q/F$ analogously as in \eqref{Eq:BRedBracket}.  
It is straightforward to conclude that the $F$-reduction of $(T^*\widetilde Q, \{\cdot , \cdot \}_{\widetilde\BB}, \widetilde H)$ and the $G$-reduction of $(\M, \{\cdot , \cdot \}_{\BB}, H_\subM)$ coincide and that is why we use the same notation: $(\M/G, \{\cdot , \cdot \}^{\BB}_\red, H_\red)$.

The properties of the bracket $\{\cdot , \cdot \}^{\BB}_\red$ on $\M/G$ are studied in \cite{BY} (see also \cite{GNM}) where it is observed that, not only it describes the dynamics: $X_\red = \{\cdot, H_\red\}_\red^\BB$, but also the bracket $\{\cdot, \cdot\}_{\red}^\BB$ is regular with an integrable characteristic distribution and hence it is twisted Poisson. This last fact is easily seen using that the rank of its characteristic distribution is $\textup{dim}(\M/G) -k$ and the $k$ functions $\bar{J}_i$ are {\it Casimirs}.  Precisely, recall that $\{\widetilde J_1,...,\widetilde J_k\}$ are the partially reduced horizontal gauge momenta with $\{\eta_1,...,\eta_k\}$ the partially reduced horizontal gauge symmetries, then Prop.~\ref{Prop:HGM-MomMap} guarantees that $\{\cdot, \widetilde J_i\}_{\widetilde\BB} = (\eta_i)_{\mbox{\tiny{$T^*\widetilde Q$}}}$ and thus $\{\cdot, \bar{J}_i\}_\red^\BB = 0$ for $\bar{J}_i\in C^\infty(\M/G)$ such that $\rho_{\mbox{\tiny{$F$}}}^*\bar{J}_i=\widetilde J_i$.

Next, we see how Theorem \ref{T:MW} characterizes also the almost symplectic structure on the foliation associated to the reduced bracket $\{\cdot, \cdot \}_{\red}^\BB$.

As usual, $\mu = c_i\mu^i$ is a $\mathfrak{f}^*$-valued function on $\widetilde Q$  where $\mu^i\in  \widetilde{\mathfrak{B}}^*_{\mbox{\tiny{HGS}}}$ and $c_i\in \R$ for $i=1,...,k$.

\begin{theorem} \label{T:LeavesOfBracket}
  The leaves of the twisted Poisson bracket $\{\cdot, \cdot\}_{\emph\red}^\BB$ on $\M/G$ are (the connected components of) the almost symplectic manifolds $(\widetilde J^{-1}(\mu)/F, \omega^\BB_\mu)$ obtained in Theorem~\ref{T:MW}, where $\widetilde J^{-1}(\mu)/F$ coincides with the common level sets of the reduced horizontal gauge momenta $\bar J_i$ on $\M/G$, i.e., $\widetilde J^{-1}(\mu)/F \simeq \cap_i \bar{J}_i^{-1}(c_i)$.  
\end{theorem}

\begin{proof}

Let $f\in C^\infty(T^*\widetilde Q/F)$. We will show that the vector field $\overline{X}: = \{\cdot, f\}_\red^\BB$ defined on $T^*\widetilde Q/F$ satisfies that, for $\bar{\alpha}\in \widetilde J^{-1}(\mu)/F$, $\overline{X}(\overline{\alpha})\in T_{\overline\alpha} (\widetilde J^{-1}(\mu)/F)$ and 
\begin{equation}\label{Eq:ProofLeaves}
{\bf i}_{\overline X({\overline\alpha})} \omega_\mu^\BB = d(\iota_\mu^\red)^* f  (\overline\alpha),
\end{equation}
where, as usual, $\iota_\mu^\red: \widetilde J^{-1}(\mu)/F \to \M/G$ is the natural inclusion.  In fact, first observe that, since $\overline{X}$ belongs to the characteristic distribution of $\{\cdot, \cdot\}_\red^\BB$ then $\overline{X}(\overline J_i) =0$ and hence $\overline{X}(\overline\alpha)\in T_{\overline\alpha}(\bar{J}_i^{-1}(c_i)) = T_{\overline\alpha}(\widetilde{J}^{-1}(\mu)/F)$, using Prop.~\ref{P:LevelSets}. Now, let $\alpha\in \widetilde J^{-1}(\mu)$ and $X(\alpha)\in T_\alpha (\widetilde J^{-1}(\mu))$ such that $\rho_F(\alpha)=\overline{\alpha}$ and $T\rho_{\mbox{\tiny{$F$}}}( X(\alpha))=\overline X(\overline{\alpha})$. 
By the definition of $\overline{X}$ and $X$, we have that $X(\alpha) = \{\cdot, \rho_{\mbox{\tiny{$F$}}}^* f\}_{\widetilde\BB}(\alpha)$ and therefore, by Prop.~\ref{Prop:BracketChap},  ${\bf i}_{X(\alpha)} \widetilde \Omega_\BB = d\rho_{\mbox{\tiny{$F$}}}^* f(\alpha)$ and hence ${\bf i}_{X(\alpha)} \iota_\mu^*\widetilde \Omega_\BB = d\iota_\mu^*\rho_{\mbox{\tiny{$F$}}}^* f(\alpha)$ which, by Theorem \ref{T:MW}, is equivalent to \eqref{Eq:ProofLeaves}.
\end{proof}

\begin{remark}\label{R:Indentification2}
\begin{enumerate}
 \item[$(i)$] As a consequence of the three main Theorems (Thms.~\ref{T:MW}, \ref{T:mu-identification} and \ref{T:LeavesOfBracket}) we conclude that the (connected components of the) almost symplectic leaves of the reduced bracket $\{\cdot, \cdot\}_\red^\BB$ on $\M/G$ are diffeomorphic to $(T^*\overline Q, \Omega_{\mbox{\tiny{$\overline Q$}}} + \widehat{\mathcal{B}}_\mu)$, see Remark~\ref{R:Indentification}.
 \item[$(ii)$] Following the notation of \cite{BY}, we have that the reduction of the manifold $(T^* \widetilde Q, \widetilde \Omega+ \widetilde{B}_1)$ gives the Poisson bracket $\{\cdot, \cdot\}_\red^1$ on $\M/G$ for which the symplectic leaves are diffeomorphic to $(T^* \overline Q, \Omega_{\mbox{\tiny{$\overline Q$}}})$.  Moreover, Theorem~\ref{T:LeavesOfBracket} puts in evidence the gauge relation of $\{\cdot, \cdot\}_\red^1$ and $\{\cdot, \cdot\}_\red^\BB$ since they have the same foliation and the 2-form on each leaf is given by $\Omega_{\mbox{\tiny{$\overline Q$}}}$ and $\Omega_{\mbox{\tiny{$\overline Q$}}} + \widehat{\mathcal{B}}_\mu$, respectively. 
\end{enumerate}
\end{remark}

\section{Examples} \label{S:Examples}

In this section we study four different examples.  All of them admit a vertical complement of the constraints satisfying the vertical symmetry condition and $k$ horizontal gauge momenta  given by the nonholonomic momentum map evaluated in non-constant sections. 

\subsection{Nonholonomic particle} \label{Ex:NHParticle}
The nonholonomic particle is the toy example describing the motion of a particle in $Q = \R^3$ with coordinates $q=(x,y,z)$ determined by the (kinetic) Lagrangian $L = \frac{1}{2}(\dot{x}^2+\dot{y}^2+\dot{z}^2)$ and the nonintegrable distribution $D=\textup{span}\left\{\partial_x+y \partial_z,\partial_y\right\}$. The translational Lie group $G=\mathbb{R}^2$ acting on the first and third variable of $\R^3$ is a symmetry of the nonholonomic system.  The system $(\M, \Omega_\subM|_\C, H_\subM)$ admits a $G$-invariant horizontal gauge momentum $J_\zeta = \tfrac{1}{\sqrt{1+y^2}}p_x$ on $\M$ (where $(p_x,p_y,p_z)$ are the coordinates on $T_q^*Q$ associated to the basis $\{dx,dy,dz-ydx\}$,  see e.g. \cite{Bloch:Book}).

{\bf First step reduction.}  The vertical complement $W=\textup{span}\{\partial_z\}$ satisfies the vertical symmetry condition and therefore, the reduction of the system by the Lie group $G_\subW = \R$ gives the partially reduced nonholonomic system $(T^*\widetilde Q, \widetilde \Omega, \widetilde H)$ where $\widetilde Q=\R^2$, 
$\widetilde\Omega = \Omega_{\mbox{\tiny{$\R^2$}}} - B_{\mbox{\tiny{$\!\langle\! J\mathcal{K}\!\rangle$}}}  = dx\wedge dp_x+dy\wedge dp_y-\displaystyle\tfrac{y}{1+y^2}p_x\,dx\wedge dy$ and $\widetilde H = \tfrac{1}{2}\left(\tfrac{1}{1+y^2}p_x^2+p_y^2\right).$
The action of the Lie group $F= G/G_\subW \simeq \R$ on $\widetilde Q$ (given by the translation on the first variable) is a symmetry of the partially reduced system on $T^*\widetilde Q$. By Prop.~\ref{Prop:tildeJ_i} the partially reduced horizontal gauge momentum is given by $\widetilde J_{\eta_1} = \tfrac{1}{\sqrt{1+y^2}}p_x$ with $\eta_1 = \tfrac{1}{\sqrt{1+y^2}}{\bf 1}$ where ${\bf 1}\in \mathfrak{f}=\R$ such that ${\bf 1}_{\mbox{\tiny{$\widetilde Q$}}} = \partial_x$. Following the computations on Sec.~\ref{Ss:B}, the 2-form defining the gauge transformation is ${\bf B} = 0$ and hence $(\eta_1)_{\mbox{\tiny{$T^*\widetilde Q$}}}$ is the hamiltonian vector field associated to the function $\widetilde J_{\eta_1}$ for the 2-form $\widetilde \Omega$ as Prop. \ref{Prop:HGM-MomMap} shows.

{\bf Momentum map and reduction.}
The canonical momentum bundle map $\widetilde J: T^*(\R^2) \to \mathfrak{f}^*$ is given, at each $(x,y)\in\R^2$,  by 
$$
\langle \widetilde J(x,y,p_x,p_y) , \eta(x,y)\rangle = f(x,y) p_x,
$$
where $\eta(x,y)=f(x,y) {\bf 1}\in \mathfrak{f}$ for $f\in C^\infty(\R^2)$.  The dual element $\mu \in \Gamma(\R^2\times\mathfrak{f}^*)$ associated to $\eta_1 \in \Gamma(\R^2\times\mathfrak{f})$ is given by $\mu = \sqrt{1+y^2}{\bf 1}^*$, where ${\bf 1}^*\in \mathfrak{f}^*$ is the dual element of ${\bf 1}$.  For $\mu_c= c \mu$ with $c\in\R$, $\widetilde J^{-1}(\mu_c)=\{  (x,y,p_x,p_y) \ : \  p_x = c\sqrt{1+y^2}\} = \widetilde J_{\eta_1}^{-1}(c)$, recovering Prop.~\ref{P:LevelSets}.  Following Theorem~\ref{T:MW}, $\iota_{\mu_c}^*\widetilde \Omega = dy\wedge dp_y$ and hence on $\widetilde J^{-1}(\mu_c)/F\simeq T^*(\R)$ the 2-form $\omega_{\mu_c} = dy\wedge dp_y$ is symplectic and coincides with the canonical 2-form on $T^*(\R)$, see Theorem~\ref{T:mu-identification}. Therefore, each symplectic leaf associated to $\{\cdot, \cdot\}_\red$ on $\M/G$ is identified with the canonical symplectic manifold $(T^*(\R), \Omega_{\mbox{\tiny{$\R$}}})$.

\subsection{Snakeboard}\label{Ex:Snakeboard}

The snakeboard describes the dynamics of a skateboard but allowing the axis of the wheels to rotate by the effect of the human rider creating a torque, so that the board spins about a vertical axis, see e.g. \cite{Bloch:Book}.
The system is modelled on the manifold $Q = \SE(2) \times S^1 \times S^1$ with coordinates $q = (\theta, x, y, \psi, \phi)$, where $(\theta, x, y)\in \SE(2)$ represents the orientation and the position of the board, $\psi$ is the angle of
the rotor with respect to the board and $\phi$ is the angle of the front and back wheels with respect to the
board (which in this simplified model they are assumed to be equal).  We denote by $m$ the mass of the board, $r$ the distance from the center of the board to the pivot point of the wheel axes and by $\mathbb{J}$, $\mathbb{J}_0$ the inertia of the rotor and of the board respectively.   
The lagrangian is given by
$$
L(q,\dot{q})=\tfrac{m}{2}(\dot{x}^2+\dot{y}^2+r^2\dot{\theta}^2)+\tfrac{1}{2} \mathbb{J}\dot{\psi}^2+\mathbb{J}\dot{\psi}\dot{\theta}+\mathbb{J}_0 \dot{\phi}^2.  
$$

In this simplified version, see \cite{BKMM}, the constraint 1-forms can be written as  $\varepsilon^x = dx+r\cos \theta  \cot \phi \ d\theta$ and $\varepsilon^y = dy + r\sin \theta  \cot \phi \ d\theta,$ with $\phi\neq 0,\pi$, and hence the constraint distribution $D$ is given by 
$$
D=\text{span} \left\{ Y_\theta := \partial_\theta - r \cos \theta  \cot \phi \ \partial_x - r  \sin \theta \cot \phi \ \partial_y , \, \partial_\psi , \, \partial_\phi  \right\}.$$

The action of the Lie group $G=\SE(2) \times S^1$ on $Q$ given, at each $q\in Q$ and $(\alpha, a ,b,\beta)\in G$, by
$$ 
\Psi_{(\alpha,a,b,\beta)} (\theta,x,y,\psi,\phi) = (\theta+\alpha,x\cos \alpha - y \sin \alpha +a, x\sin \alpha+y\cos\alpha+b,\psi+\beta,\phi),
$$
is free and proper and it defines a symmetry for the nonholonomic system.  Then $V=\textup{span}\{Y_{\theta},\partial_{\psi},\partial_{x},\partial_y\}$.

{\bf The manifold $\M$}. Observe that $W=\text{span} \left\{ \partial_x,\partial_y \right\}$ is a vertical complement of the constraints, and we may consider the adapted basis $\mathfrak{B}_{\mbox{\tiny{$TQ$}}} = \left\{ Y_\theta, \partial_\psi , \partial_\phi, \partial_x,\partial_y \right\}$ of $TQ=D\oplus W$ with dual basis $\mathfrak{B}_{\mbox{\tiny{$T^*\!Q$}}} = \left\{ d\theta, d\psi , d\phi, \varepsilon^x, \varepsilon^y \right\}$.  If we denote by $(p_\theta,p_\psi,p_\phi,p_x,p_y)$ the associated coordinates on $T^*Q$, then the manifold $\M\subset T^*Q$ is given by $p_x= -\cos \theta \, F(\phi) (p_\theta-p_\psi)$ and $p_y = - \sin \theta \, F(\phi) (p_\theta-p_\psi)$, where $F(\phi) = mr\tfrac{\sin\phi \cos\phi }{m r^2 - J \sin^2 \phi }$.

{\bf The vertical symmetry condition and first step reduction}.  The vertical complement $W$ of the constraints $D$ satisfies the vertical symmetry condition, that is, $G_\subW=\R^2$. Then the system is $\R^2$-Chaplygin and the partially reduced nonholonomic system takes place in $T^*\widetilde Q$ for $\widetilde{Q} \simeq \SO(2) \times S^1 \times S^1$ with coordinates $(\theta,\psi,\phi)$ and it is determined by the partially reduced hamiltonian $\widetilde H$ and the 2-form
\begin{equation}\label{Eq:Snake1}
\widetilde \Omega = \Omega_{\mbox{\tiny{$\widetilde{Q}$}}} - B_{\mbox{\tiny{$\!\langle\! J\mathcal{K}\!\rangle$}}} = -d(p_\theta d\theta+p_\psi d\psi + p_\phi d\phi) + r\tfrac{F(\phi)}{\sin^2\phi}(p_\theta-p_\psi) d\theta\wedge d\phi,
\end{equation}
where we recall that $\rho_\subW^*B_{\mbox{\tiny{$\!\langle\! J\mathcal{K}\!\rangle$}}} = \langle J,\mathcal{K}_\subW \rangle$ and $\langle J,\mathcal{K}_\subW \rangle |_\C = (p_x d\varepsilon^x + p_y d\varepsilon^y) |_\C$. 

The remaining symmetry is given by the action of the Lie group $F=G/G_\subW\simeq S^1 \times S^1$ on $\widetilde{Q}$ given, at each $(\alpha,\beta) \in F$ and $(\theta,\psi,\phi)\in \widetilde Q$, by
$\widetilde \Psi_{(\alpha,\beta)} (\theta,\psi,\phi) = (\theta+\alpha,\psi+\beta,\phi),$ with trivial Lie algebra $\mathfrak{f}\simeq \R^2$.

{\bf Horizontal gauge momenta and gauge transformation}. First observe that $S=\textup{span}\{Y_\theta, \partial_\psi\}$.  Following \cite{BalSan}, we have also two ($G$-invariant) horizontal gauge momenta ${J}_1, {J}_2$ given by 
\[ 
{J}_1=\exp \left( r \int \tfrac{F(\phi)}{\sin^2(\phi)} \ d\phi \right) (p_\theta-p_\psi)   \hspace{0.4cm} \text{ and } \hspace{0.4cm} {J}_2=p_\psi.  
\]

By the $G_\subW$-invariance of ${J}_1$ and $J_2$,  we can define the partially reduced horizontal gauge momenta $\widetilde J_1 = \langle \widetilde J, \eta_1\rangle$ and $\widetilde J_2= \langle \widetilde J, \eta_1\rangle$ on $T^*\widetilde Q$ as in \eqref{Eq:TildeJi} 
for $\widetilde J: T^*\widetilde Q\to \mathfrak{f}^*$ the canonical momentum map and $\eta_1, \eta_2 \in \Gamma(\widetilde Q \times \mathfrak{f})$ with $\eta_1= E(\phi)(e_1-e_2)$ and $\eta_2=e_2$, where $E(\phi)=\exp\left( r\int\tfrac{F(\phi)}{\sin^2\phi} \ d\phi \right)$ and $\{e_1,e_2\}$ is the canonical basis of $\mathfrak{f}$, see Lemma \ref{L:lemma1-1} and Sec.~\ref{Ss:CanMomentMap}. 

From Section~\ref{Ss:B} (for details see \cite{BalSan}) we can see that ${\bf B}=0$ and then Prop.~\ref{Prop:HGM-MomMap} is verified directly for the 2-form $\widetilde \Omega$. 

{\bf Momentum map and reduction}.
Let us consider the canonical momentum map $\widetilde J: T^*\widetilde Q\to \mathfrak{f}^*$. Following \eqref{Eq:RedBasisHGS} and \eqref{Eq:RedDualBasisHGS} we have the basis  of sections $\widetilde{\mathfrak{B}}_{\mbox{\tiny{HGS}}}=\{\eta_1, \eta_2\}$ on $\widetilde Q \times \mathfrak{f}\to \widetilde Q$ and  its dual basis  $\widetilde{\mathfrak{B}}^*_{\mbox{\tiny{HGS}}}=\{\mu^1, \mu^2\}$ so that $\mu^1 = \frac{1}{E(\phi)}e^1$ and $\mu^2=e^1+e^2$, where $\{e^1,e^2\}$ is the dual basis of $\mathfrak{f}^* $ associated to $\{e_1,e_2\}$.  
 Therefore for $\mu=c_1\mu^1+c_2\mu^2 \in \Gamma(\widetilde Q \times \mathfrak{f}^*)$ with $c_1,c_2 \in \mathbb{R}$, we obtain that
\begin{equation}\label{Eq:Snake2}
\widetilde{J}^{-1}(\mu)  = \widetilde{J}_1^{-1}(c_1) \cap \widetilde{J}_2^{-1}(c_2)  = \{ (\theta,\phi,\psi,p_\theta,p_\phi,p_\psi) \in T^*\widetilde{Q} :  p_\theta = \tfrac{c_1}{E(\phi)}+c_2 \ \text{ and }  \ p_\psi=c_2 \}.
\end{equation}

From \eqref{Eq:Snake1} and \eqref{Eq:Snake2} we can compute $\iota_\mu^*\widetilde \Omega$ and hence, following Theorem \ref{T:MW} we conclude that the (almost) symplectic form on $\widetilde{J}^{-1}(\mu)/F$ is $\omega_\mu = d\phi\wedge dp_\phi$.  Therefore, in agreement with Theorem~\ref{T:mu-identification} and Theorem~\ref{T:LeavesOfBracket}, we identify 
the leaves of $\{\cdot, \cdot\}_\red$ defined in \eqref{Eq:RedBracket} (and computed explicitly in \cite{BalSan}) with the canonical symplectic manifolds $(T^*S^1, \Omega_{\mbox{\tiny{$S^1$}}})$ (observe that $\textup{dim}(\overline{Q})=1$). 
\subsection{Chaplygin ball}

Next, we study the celebrated example called the Chaplygin ball in the context of \cite{BorMa2001,Duistermaat, GN2008}, see also \cite{Hoch}. Consider a ball of radius $r$ with an inhomogeneous mass distribution (but with the center of mass coinciding with the geometric center) that rolls without sliding on a plane. We denote by $\mathbb{I}$ the inertia tensor that is represented as a diagonal matrix with positive entries given by the principal moments of inertia $\mathbb{I}_1$, $\mathbb{I}_2$, $\mathbb{I}_3$.  The configuration manifold is $Q=\SO(3)\times \R^2$, where the rotational matrix $g\in \SO(3)$ represents the orientation of the ball relating the orthogonal frame attached to the body with the one fixed in space and $(x,y)\in \R^2$ represents the position of the center of mass of the ball.  The Lagrangian is just the kinetic energy 
$$
L((g,x,y), (\vecOm, \dot x, \dot y))= \tfrac{1}{2} \langle \mathbb{I}\vecOm, \vecOm \rangle +   \tfrac{m}{2}(\dot x^2 + \dot y^2),
$$
where $\vecOm = (\Omega_1, \Omega_2, \Omega_3)$ is the angular velocity in body coordinates and $m$ is the total mass of the ball. 

The nonholonomic constraints are written as $\dot x = r \langle \vecbeta, \vecOm\rangle$ and $\dot y = - r\langle \vecalpha, \vecOm \rangle$ defining the constraints 1-forms given by 
$$
\epsilon^x = dx - r\langle \vecbeta, \vecL\rangle \qquad \mbox{and} \qquad \epsilon^y = dy+ r\langle \vecalpha, \vecL \rangle,
$$
where $\vecalpha = (\alpha_1, \alpha_2, \alpha_3)$ and $\vecbeta=(\beta_1, \beta_2, \beta_3)$ are the first and second rows of the matrix $g$, $\vecL = (\lambda_1, \lambda_2, \lambda_3)$ are the left-invariant Maurer Cartan 1-forms on $\SO(3)$ and $\langle \cdot, \cdot \rangle$ is the natural pairing in $\R^3$.  

The system has a symmetry given by the action $\Psi$ of the Lie group $G=\SE(2)$ on $Q$ given, at each $(h,(a,b))\in \SE(2)$ and $(g,(x,y))\in \SO(3)\times \R^2$,  by 
$
\Psi_{(h,(a,b))} (g,(x,y)) = ( \tilde h  g, h(x,y)^t + (a,b)^t ),
$
where $\widetilde h = \left( \! \begin{smallmatrix*}[c]  h & 0 \\ 0 & 1  \end{smallmatrix*} \!\right) \in \SO(3)$ and $(\cdot)^t$ is the transpose of the element $(\cdot)$. 

Let us denote by $\{X_1^L, X_2^L, X_3^L\}$ the left invariant vector fields on $\SO(3)$ (so that at the identity they are align with the canonical basis of $\mathfrak{so}(3)$) dual to $\{\lambda_1, \lambda_2, \lambda_3\}$ and then we observe that the vector fields $X_i := X_i^L +r\beta_i \partial_x - r\alpha_i\partial_y$ for $i=1,2,3$ generate the constraint distribution $D$. Observe also that the vertical distribution $V$ is generated by $\{\langle \vecgamma, {\bf X}\rangle, \partial_x, \partial_y\}$, where $\vecgamma = (\gamma_1, \gamma_2, \gamma_3)$ is the third row of the matrix $g$ and ${\bf X}=(X_1,X_2,X_3)$. Therefore, the distribution $W=\{\partial_x, \partial_y\}$ is a vertical complement of the constraints and hence we may consider the adapted basis of $TQ$ given by $\mathfrak{B}_{\mbox{\tiny{$TQ$}}} =\{ X_1, X_2,X_3,\partial_x,\partial_y\}$ with dual basis $\mathfrak{B}_{\mbox{\tiny{$T^*\!Q$}}} = \{\lambda_1, \lambda_2,\lambda_3, \epsilon^x, \epsilon^y\}$.  

{\bf The manifold $\M$}.  If $(M_1,M_2,M_3,p_x,p_y)$ are the coordinates associated to the basis $\mathfrak{B}_{\mbox{\tiny{$T^*\!Q$}}}$ then the manifold $\M\subset T^*Q$ is given by $p_x = mr\langle \vecbeta, \vecOm \rangle$ and $p_y = - mr \langle \vecalpha, \vecOm \rangle$,
where ${\bf M}=(M_1,M_2,M_3)$ and $\vecOm$ are related by ${\bf M} = \mathbb{I}\vecOm + mr^2 \langle \vecgamma, \vecOm \rangle \vecgamma$.

{\bf The vertical symmetry condition and first step reduction.}   The Lie group $G_\subW = \R^2$ is a closed normal subgroup of $G$ and hence it is a symmetry of the nonholonomic system that makes it into a Chaplygin system:  $TQ = D\oplus W$ for $W =\textup{span}\{\partial_x, \partial_y\}$. 
Since $\widetilde Q \simeq \SO(3)$ then $\M/G_\subW \simeq T^*\SO(3)$ with coordinates $(g,{\bf M})$.  The 2-form $\widetilde \Omega$ on $T^*\SO(3)$, defining the partially reduced dynamics, is given by 
\begin{equation*}
\widetilde \Omega= \Omega_{\mbox{\tiny{$\widetilde Q$}}} - B_{\mbox{\tiny{$\!\langle\! J\mathcal{K}\!\rangle$}}}  = -d(M_i \lambda_i) -  mr^2( \vecOm  -   \langle \vecgamma,\vecOm\rangle \vecgamma) d\vecL.
\end{equation*}
 The remaining Lie group $F= G/G_\subW\simeq S^1$ leaves invariant the system $(T^*\SO(3),\widetilde \Omega,\widetilde H)$.

{\bf Horizontal gauge momentum and gauge transformation.}
Following \cite{Duistermaat} (see also \cite{BorMa2001,GN2008}), this example admits one $G$-invariant horizontal gauge momentum $J_\zeta\in C^\infty(\M)$ defined by the section $\zeta = ( 1,-x,y)$ of $\g_{\mbox{\tiny{$S$}}}$, that is, $J_\zeta = \langle J^\nh, (1,-x,y)\rangle = \langle \vecgamma, {\bf M}\rangle$. 
Moreover, since $J_\zeta$ is basic, it descends to a partially horizontal gauge momentum $\widetilde J_1 = \langle \widetilde J, {\bf 1}\rangle = {\bf i}_{\widetilde Y} \Theta_{\mbox{\tiny{$\widetilde Q$}}}$ where $\widetilde Y = {\bf 1}_{\mbox{\tiny{$\widetilde Q$}}} = \langle \vecgamma, {\bf X}\rangle$ for ${\bf 1}\in \mathfrak{f}\simeq \R$, see Lemma \ref{L:lemma1-1}. 

Since $\textup{rank}(S) =\textup{dim}(\mathfrak{f})=1$, following Sec.~\ref{Ss:B} (see also \cite{BY,GN2008}), the system admits a dynamical gauge transformation by the 2-form  ${\bf B}  = m r^2 \langle \vecOm, d\vecL\rangle$ so that $\widetilde Y ={\bf 1}_{\mbox{\tiny{$\widetilde Q$}}}$ is the hamiltonian vector field associated to $\widetilde J_1$ with respect to the 2-form  
\begin{equation}\label{Bola:omegaBtilde}
\widetilde \Omega_\BB = \Omega_{\mbox{\tiny{$\widetilde Q$}}} - B_{\mbox{\tiny{$\!\langle\! J\mathcal{K}\!\rangle$}}} +\widetilde{\bf B} = -d(M_i \lambda_i) + mr^2\langle \vecgamma,\vecOm\rangle\langle \vecgamma, d\vecL\rangle.
\end{equation}
 Moreover, from Sec.~\ref{Ss:B},  the 2-form ${\bf B}$ is written as ${\bf B} =B_1+\mathcal{B}$ where $B_1 = \langle J,\mathcal{K}_\subW\rangle - J_\zeta \, d\langle \vecgamma, \vecL\rangle$ and $\mathcal{B}=  mr^2\langle \vecgamma,\vecOm\rangle\langle \vecgamma, d\vecL\rangle + J_\zeta \, \Phi_{S^2}$ for 
 $\Phi_{S^2} = d\langle \vecgamma, \vecL\rangle= \gamma_1\, d\gamma_2\wedge d\gamma_3+ \gamma_2\, d\gamma_3\wedge d\gamma_1 + \gamma_3\, d\gamma_1\wedge d\gamma_2$, see \cite{BGN,BY} for details. 

{\bf Momentum map and reduction}. From \eqref{Bola:omegaBtilde} we observe that $\widetilde{\Omega}_\BB$ is the sum of the canonical 2-form on $T^*\widetilde Q$ and a basic 2-form (with respect to the principal bundle $T^*\SO(3) \to T^*\SO(3)/F$), which implies that the canonical momentum map $\widetilde J:T^* \SO(3) \to \mathfrak{f}^*$ is a standard momentum map for $\widetilde \Omega_\BB$. Moreover, since $\widetilde J_1 =\langle \widetilde J,{\bf 1}\rangle$  we can perform a standard (almost) symplectic reduction (as it was done in \cite{BalFern}). 

In what follows we compute the almost symplectic foliation defined in Theorem~\ref{T:MW} in order to illustrate Theorem~\ref{T:mu-identification} and Theorem~\ref{T:LeavesOfBracket} and enlighten the almost symplectic leaves associated to the twisted Poisson bracket described in \cite{BGN,BorMa2001,GN2008}.  For $\gamma_3\neq 0$, consider the basis $\mathfrak{B}_{\mbox{\tiny{$T\widetilde Q$}}} = \{\widetilde Y, \widetilde X_i= X_i^L - \gamma_i \widetilde Y\}$ for $i=1,2$ with the dual basis $\mathfrak{B}_{\mbox{\tiny{$T^*\!\widetilde Q$}}} = \{\epsilon^{\widetilde Y} = \langle \vecgamma, \vecL\rangle, \widetilde X^1 = \gamma_3^{-1} (-\gamma_1\lambda_3+\gamma_3\lambda_1), \widetilde X^2 = \gamma_3^{-1} (-\gamma_2\lambda_3+\gamma_3\lambda_2)\}$. 
If $(\widetilde p, p_1,p_2)$ are the coordinates on $T^*\widetilde Q$ associated with this basis, then 
$$
\widetilde\Omega_\BB = -d(p_1\widetilde X^1 +p_2\widetilde X^2) - d(\widetilde p\, \epsilon^{\widetilde Y}) + mr^2\langle \vecgamma,\vecOm\rangle\langle \vecgamma, d\vecL\rangle.
$$
Then,  using Remark~\ref{R:G-invJ} we have that $\widetilde J^{-1}(\mu)/F = \bar{J}_1^{-1}(c) = \{(\vecgamma, \widetilde p, p_1,p_2) \ : \ \langle \vecgamma, \vecgamma\rangle = 1, \widetilde p=c\}\simeq T^*(S^2)$, 
and hence  
$$
\omega_\mu^\BB = \Omega_{S^2} + (c-mr^2\langle \vecgamma, \vecOm\rangle ) \Phi_{S^2},
$$
where $\Omega_{S^2}$ is the canonical 2-form on $T^*(S^2)$.  Therefore, the almost symplectic foliation of the bracket $\{\cdot, \cdot\}_\red^\BB$ (computed explicitly in \cite{BorMa2001,GN2008}) is identified with the manifolds $(T^*(S^2), \Omega_{S^2} + \widehat{\mathcal{B}}_\mu)$, where  $\widehat{\mathcal{B}}_\mu$ is the 2-form on $T^*(S^2)$ given by $\widehat{\mathcal{B}}_\mu  = (c-mr^2\langle \vecgamma, \vecOm\rangle ) \Phi_{S^2}$ (observe also that $\overline{\mathcal{B}} = (\bar{J}_1 -mr^2 \langle \vecgamma, \vecOm\rangle) \Phi_{S^2}$). In this case, we can compute a conformal factor $f_\mu$ for each leaf $(T^*(S^2), \Omega_{S^2} + \widehat{\mathcal{B}}_\mu)$ in order to obtain the conformal factor for the bracket $\{\cdot, \cdot\}_\red^\BB$ (see e.g., \cite{BalFern}).

\subsection{Solids of revolution on a plane}\label{Ex:Solids}
Let us consider a convex body of revolution, i.e., a body that is geometrically and dynamically symme\-tric under rotations about a given axis, rolling on a plane without sliding.  This example is interesting because the horizontal gauge momenta cannot be explicitly written. However, we will see that the reduction of Theorem \ref{T:MW} gives a symplectic foliation, where the nonholonomic dynamics lives, that is diffeomorphic to  $(T^*S^1,\Omega_{\mbox{\tiny{$S^1$}}})$ as Theorem \ref{T:mu-identification} asserts. 

 For this example we follow \cite{Cushman1998,Book:CDS} and we keep the notation and framework of the previous example. We assume that the body is invariant under rotations around ${\bf e}_3$ and hence the principal moments of inertia  are $I_1=I_2$ and $I_3$.  The total mass of the body is $m$ and the position of the center of mass is represented by the coordinates $\vecx  = (x,y,z) \in \R^3$ while the relative position of the body is given by the rotational matrix $g \in \SO(3)$. The lagrangian $L : T(\SO(3)\times \R^3 ) \to \R$ is of mechanical type and is given by
$$
L((g, \vecx),( \vecOm, \dot\vecx)) = \frac{1}{2} \langle \mathbb{I} \vecOm, \vecOm\rangle + \frac{1}{2} m \langle \dot \vecx, \dot \vecx\rangle - m{\bf g} \langle \vecx, {\bf e}_3 \rangle,
$$
where $\mathbb{I}$ is the inertia matrix with entries $I_1,I_2,I_3$,  $\vecOm = (\Omega_1, \Omega_2, \Omega_3)$ is the angular velocity of the body in body coordinates, ${\bf g}$ is the constant of gravity and $\langle \cdot, \cdot \rangle$ denotes the inner product in $\R^3$.

Let $\vecs$ be the vector from $\vecx$ to a fixed point on the surface $\mathcal{S}$ of the body. If we denote by $\vecgamma = (\gamma_1, \gamma_2, \gamma_3)$ the third row of the matrix $g \in \SO(3)$, then $\vecs$ can be represented by the map $\vecs : S^2 \to \mathcal{S}$ so that $\vecs(\vecgamma) = (\varrho(\gamma_3)\gamma_1 , \varrho(\gamma_3) \gamma_2 , \zeta(\gamma_3 ))$,
where $\varrho= \varrho (\gamma_3)$ and $\zeta=\zeta(\gamma_3)$ are the smooth functions defined in \cite[Chap.6.7]{Book:CDS} that depend on the shape of the body. Throughout this work, we will also denote $\vecs(\vecgamma) = \varrho  \vecgamma -L {\bf e}_3$ where $L=L(\gamma_3) = \varrho  \gamma_3 - \zeta$. Since the body rolls on a plane, the configuration space $Q$ is diffeomorphic to $\SO(3)\times \R^2$ and it is described by
$$ 
Q = \{(g, \vecx) \in \SO(3) \times \R^3 : z = -\langle \vecgamma, \vecs\rangle\}. 
$$

The nonholonomic constraints describing the rolling without sliding are written as $
g^t \dot \vecx= - \vecOm \times \vecs.
$
Using that $(\vecOm, \dot\vecx)$ are the coordinates associated to the basis $\{X_1^L, X_2^L, X_3^L, \partial_x, \partial_y, \partial_z\}$ of $TQ$, for $X_i^L$ the left invariant vector fields on $\SO(3)$, we conclude that the constraint distribution $D$ on $Q$ can be written as $D= \textup{span}\{X_1, X_2, X_3\}$, where the vector fields $X_i$ are defined to be 
$$
X_i := X_i^L + (\vecalpha \times \vecs)_i\partial_{x} + (\vecbeta \times \vecs)_i \partial_{y} + (\vecgamma \times \vecs)_i \partial_{z}, \mbox{ \ \ for } i=1,2,3,
$$
with $\vecalpha$ and $\vecbeta$ the first and second rows of the matrix $g$. Therefore, the constraints 1-forms are
$$
\epsilon^1 = dx - \langle \vecalpha , \vecs \times \vecL \rangle \qquad \mbox{and} \qquad \epsilon^2 = dy - \langle \vecbeta, \vecs \times \vecL \rangle,
$$
where $\vecL = (\lambda_1 , \lambda_2 , \lambda_3 )$ are the (Maurer-Cartan) 1-forms on $\SO(3)$ dual to the left invariant vector fields $\{X_1^L , X_2^L , X_3^L \}$.

For $(g,(x,y))$ coordinates on $Q\simeq \SO(3)\times\R^2$, we define the action of the Lie group $G= S^1\times \textup{SE}(2)$  on $Q$ given, at each $(h_1, (h_2, (a,b)))\in G$, by 
$$
\Psi_{(h_1, (h_2, (a,b)))} (g, (x,y)) =  ( \tilde h_2g\tilde h_1^{-1}, h_1h_2(x,y)^t +(a,b)^t),
$$
where $h_1$ and $h_2 \in \SO(2)$ are orthogonal $2\times 2$ matrices and $\tilde h_i = \left(\begin{array}{cc} h_i & 0 \\ 0 & 1\end{array} \right) \in \SO(3)$. Since this action is not free, from now on, we consider the manifold $Q$ given by the coordinates $(g,(x,y))$ with $\gamma_3 \neq \pm 1$ and hence with this restriction it defines a  symmetry of this nonholonomic system as in \cite{BalSan}.

The Lie algebra associated to $G$ is $\g \simeq \R \times \R \times \R^2$.  
It is straightforward to check that the {\it dimension assumption} is satisfied since $\partial_{x},\partial_y$ are vector fields in $V$ and $D +V=TQ$. Moreover, $S=D\cap V =\textup{span}\{Y_1 := X_3 , Y_2 := \langle \vecgamma, {\bf X}\rangle \},$ with ${\bf X} = (X_1 , X_2 , X_3 )$.  The sections $\xi_1 := (1; 0, (y + \varrho\beta_3 , -x - \varrho \alpha_3 ))$ and $\xi_2 := (0; 1, (y - L\beta_3 , -x + L\alpha_3 ))$ on $Q\times \g \to Q$, verify that 
$(\xi_1)_\subQ = -X_3$ and $(\xi_2)_\subQ = \langle \vecgamma, {\bf X}\rangle$ and hence they are a basis of the bundle $\g_\subS \to Q$.

{\bf The vertical symmetry condition and first step reduction.}
If we choose the vertical complement of the constraints $W=\textup{span}\{\partial_{x}, \partial_{y}\}$ then $W$ is $G$-invariant and generated by the Lie algebra $\mathfrak{w} = \R^2$.  Therefore, $W$ satisfies the vertical symmetry condition and we can perform a reduction by the Lie group $G_\subW=\R^2$ obtaining the partially reduced nonholonomic system on $T^*\SO(3)$ with the 2-form $\widetilde \Omega = \Omega_{\mbox{\tiny{$\SO(3)$}}} - B_{\mbox{\tiny{$\!\langle\! J\mathcal{K}\!\rangle$}}}$. 
The action of the Lie group $F\simeq S^1\times S^1$ on $T^*\SO(3)$ is given, at each $(g,{\bf M})\in T^*\SO(3)$ and $(h_1,h_2)\in S^1\times S^1$, by
$$
\widetilde \Psi_{(h_1,h_2)} (g,{\bf M}) = (\tilde h_2g \tilde h_1^{-1}, \tilde h_1{\bf M}).
$$
This $F$-action defines a symmetry of the partially reduced nonholonomic system $(T^*\SO(3), \widetilde\Omega, \widetilde H)$ for $\gamma_3 \neq \pm 1$.

{\bf Horizontal gauge momenta and gauge transformation.}
Following \cite{BorMa2002b,Book:CDS} (see also \cite{BalSan}), this example admits two $G$-invariant horizontal gauge momenta $J_1, J_2$ on $\M$ defined by two sections $\zeta_1, \zeta_2$ on $Q\times \g \to Q$ given by 
$$
J_i = f_i^1 (q) J_{\xi_1} + f_i^2(q) J_{\xi_2}  \qquad \mbox{and} \qquad \zeta_i = f_i^1 (q) \xi_1 + f_i^2(q) \xi_2
$$
where $f_i^j$ are $G$-invariant functions on $Q$ (i.e., they depend only on the variable $\gamma_3$) and $J_{\xi_i} = \langle J^\nh, \xi_i\rangle$ for $i=1,2$.  The functions $f_i^j$ cannot be explicitly written, instead they are defined as a solution of an ordinary linear system of differential equations \cite{BalSan,Book:CDS}.

The horizontal gauge momenta $J_1,J_2$ descend to partially reduced horizontal gauge momenta $\widetilde J_1, \widetilde J_2$ that are given by $\widetilde J_i = \langle \widetilde J, \eta_i \rangle$ where $\eta_i = \varrho_{\g} (\zeta_i)$ are the corresponding $\mathfrak{f}$-valued functions (as in Lemma~\ref{L:lemma1-1}). That is, $\widetilde{\mathfrak{B}}_{{\mbox{\tiny{HGS}}}}= \{\eta_1, \eta_2\}$ is given by
\begin{equation}\label{Solid:HGS}
\eta_i= \tilde f_i^1{\bf e}_1 + \tilde f_i^2{\bf e}_2, 
\end{equation}
where $\tilde f_i^j$ are the functions on $\SO(3)$ such that $\rho_{\mbox{\tiny{$\widetilde Q$}}}^*\tilde f_i^j = f_i^j$ and ${\bf e}_1 = (1,0) = \varrho_{\g}(\xi_1)$, ${\bf e}_2=(0,1) = \varrho_{\g}(\xi_2)$.

Following Sec.~\ref{Ss:B} (see \cite{BY, GNM}) and Prop.~\ref{Prop:HGM-MomMap}, the partially reduced system admits a dynamical gauge transformation by a 2-form $\widetilde{\bf B}$, given by 
$$
\widetilde{\bf B} = m \varrho \langle \vecgamma , \vecs \rangle \langle \vecOm, d\vecL\rangle,
$$
 so that the infinitesimal generators associated to $\eta_1, \eta_2$ are hamiltonian vector fields of $\widetilde J_1$, $\widetilde J_2$ respectively with respect to the 2-form $\widetilde \Omega_\BB = \widetilde{\Omega} +\widetilde{\bf B}$.

{\bf Momentum map and reduction}. 
On $\widetilde Q=\SO(3)$ we will work with the basis $\mathfrak{B}_{\mbox{\tiny{$T\,\SO(3)$}}}= \{X_0= \gamma_1 X_2 - \gamma_2X_1, \widetilde{Y}_1 := X_3 , \widetilde{Y}_2 := \langle \vecgamma, {\bf X}\rangle \}$, its dual basis $\mathfrak{B}_{\mbox{\tiny{$T^*\!\SO(3)$}}}= \{X^0, \widetilde{Y}^1, \widetilde{Y}^2\}$ and the associated coordinates $(g, p_0,p_1,p_2)$ on $T^*\SO(3)$.

Let us now consider the momentum map $\widetilde J:T^*\SO(3) \to \mathfrak{f}^*$ so that, for each $\mathfrak{f}$-valued function on $\SO(3)$  such that $\eta = h_1{\bf e}_1 + h_2{\bf e}_2$ with $h_i\in C^\infty(\SO(3))$, we have that 
$$
\langle \widetilde J, \eta \rangle =  {\bf i}_{\eta_{\mbox{\tiny{$T^*\!\SO(3)$}}}} \Theta_{\mbox{\tiny{$\SO(3)$}}} = h_1p_1 + h_2p_2.
$$

Now, let $\{\mu^1, \mu^2\}$ be $\mathfrak{f}^*$-valued functions dual to ${\eta_1,\eta_2}$ defined in \eqref{Solid:HGS}. If $\mu = c_1\mu^1 + c_2\mu^2$ for $c_1, c_2 \in \R$ then  $\widetilde J^{-1}(\mu) = \{(g, p_0,p_1,p_2) \ : \  \widetilde f_j^ip_i=c_j\}$
 and hence the manifold $\widetilde J^{-1}(\mu)$ is determined by the coordinates $(g,p_0)$.  The quotient by the action of the Lie group $F$ defines the manifold $\widetilde J^{-1}(\mu)/F$ that is diffeomorphic to $T^*S^1$ and given by the coordinates $(\gamma_3, p_0)$.  Following Theorem \ref{T:MW} and since $\iota_\mu^* \widetilde\Omega_\BB = \iota_\mu^*(\Omega_{\mbox{\tiny{$\widetilde Q$}}} + \widetilde J_i d\widetilde{\mathcal{Y}}^i )$, we obtain that
  $$
 \omega_\mu^\BB = X^0\wedge dp_0,
 $$ 
 which is the canonical 2-form on $T^*S^1$ for $X^0 = \frac{d\gamma_3}{1-\gamma_3^2}$, recalling that $\gamma_3 \neq \pm 1$.

Therefore, as a consequence of Theorem~\ref{T:mu-identification} and Theorem~\ref{T:LeavesOfBracket}, we conclude that the reduced almost Poisson manifold $(\M, \{\cdot, \cdot\}_\red^\BB)$, given in Sec.~\ref{S:NHBracket}, is in fact a Poisson manifold (as it was shown in \cite{BalSolids,GNM,Ramos-04}) with symplectic leaves symplectomorphic to $(T^*S^1, \Omega_{\mbox{\tiny{$S^1$}}})$.

\end{document}